\documentclass[11pt]{amsart}

\usepackage{amssymb}
\usepackage{hyperref}
\usepackage[vmargin=1.0in, hmargin=0.7in]{geometry}

\makeatletter
\usepackage{enumitem}

\usepackage{epsfig} 
\usepackage{epic,eepic}  

\newtheorem{thm}{Theorem}[section]
\newtheorem{prop}{Proposition}[section]
\newtheorem{lemma}{Lemma}[section]
\newtheorem{cor}{Corollary}[section]

\theoremstyle{definition}
\newtheorem{definition}{Definition}[section]

\theoremstyle{remark}
\newtheorem{remark}{Remark}[section]

\numberwithin{equation}{section}

\DeclareMathOperator{\End}{End} 
\DeclareMathOperator{\Tr}{Tr} 
\DeclareMathOperator{\id}{id} 
\DeclareMathOperator{\spa}{span} 

\newcommand{\authorfootnotes}{\renewcommand\thefootnote{\@fnsymbol\c@footnote}}%

\begin{document}

\title{Some oscillatory representations of fuzzy conformal group $SU(2,2)$ with positive energy}
\maketitle

\begin{center}
 \authorfootnotes
  Samuel Bezn\'ak\footnote{\href{mailto:samuel.beznak@fmph.uniba.sk}{samuel.beznak@fmph.uniba.sk}}\textsuperscript{1}, Peter Pre\v snajder\footnote{\href{mailto:presnajder@fmph.uniba.sk}{presnajder@fmph.uniba.sk}}\textsuperscript{1}
   \par \bigskip

  \tiny \textsuperscript{1} Department of Theoretical Physics, Faculty of Mathematics, Physics and Informatics, Comenius University Bratislava \\ Mlynsk\'a dolina F1, Bratislava 842 48, Slovakia
  \par \bigskip

\end{center}

\normalsize

\begin{abstract}
We construct the relativistic fuzzy space as a non-commutative algebra of functions with purely structural and abstract coordinates being the creaction and annihilation (C/A) operators acting on a Hilbert space $\mathcal{H}_F$. Using these oscillators, we represent the conformal algebra $su(2,2)$ (containing the operators describing physical observables, that generate boosts, rotations, spatial and conformal translations, and dilatation) by operators acting on such functions and reconstruct an auxiliary Hilbert space $\mathcal{H}_A$ to describe this action. We then analyze states on such space and prove them to be boost-invariant. Eventually, we construct two classes of irreducible representations of $su(2,2)$ algebra with \textit{half-integer} dimension $d$ (\cite{Ma}): (i) the classical fuzzy massless fields as a doubleton representation of the $su(2,2)$ constructed from one set of C/A operators in fundamental or unitary inequivalent dual representation and (ii) classical fuzzy massive fields as a direct product of two doubleton representations constructed from two sets of C/A operators that are in the fundamental and  dual representation of the algebra respectively. 
\end{abstract}

\bigskip

\tableofcontents

\bigskip

\newpage

\section{Motivation and Introduction}
\setcounter{footnote}{0}
The group of the conformal transformations in four dimensions, $SO(2,4)$ (or its double cover $Spin(2,4) \cong SU(2,2)$), appears in many physical contexts. Most prominently it is an isometry group of the $AdS_5$ space which occurs as a (part of the) background for number of theories, e. g. the IIB superstring theory on $AdS_5 \times S^5$ background (\cite{MeTs}). Its compactification on $S^5$ produces KK modes governed by the unitary irreps of the $\mathcal{N}=4$ superconformal group $SU(2,2\vert 4)$, whose even part $SU(2,2) \otimes SU(4) \cong SO(2,4) \otimes SO(6)$ is the isometry of the $AdS_5 \times S^5$. Moreover, the CFT dual of this theory on its boundary is the $\mathcal{N}=4$ supersymmetric Yang -- Mills theory, which enhances the motivation for analysis of such theory. 

Within current efforts to the unifying theory of quantum gravity, the connection of the gravitation and quantum field theory requires implementation of non-trivial structure, i. e. fundamental length at the Planck scale, which was already proposed quite long time ago (\cite{Sn}, renewed interest \cite{DoFr}). Such structure implies, that we cannot distinguish space-time events with minor separation. One way of dealing with such obstacle is to construct the theory on a non-commutative (NC) space (\cite{Co1}, \cite{Sz}), that disposes of such structure without spoiling the initial symmetry of the space. 

There are several methods of how to realize NC spaces. Among the most common are matrix models (\cite{Mad1}) with the connection to the construction of the M-theory (\cite{Co2}), or in the string theory with quantum field theory on NC space as an low energy limit of an open string theory (\cite{SeWi}). They appear naturally, when the NC space is given as a non-commutative asociative algebra of functions on the (formerly commutative) space. Field theories on such fuzzy spaces can provide valuable information about the space itself (e. g. \cite{ŠuTe}, \cite{Te}, \cite{AsFi}).  However, within the ambition of quantization of a theory on such NC space, another method seems to be more convenient: the oscillator construction, already used for analyses in NC quantum mechanics (\cite{GaKoPP}, \cite{KoPP}, idea inspired by \cite{Sch}). 

The fundamental role of the pair of creation and annihilation operators (C/A)  in quantum field theories is evident: due to  Dirac, Wigner or Weinberg (\cite{We1}) the C/A's acting on a Hilbert space of states are a natural construction. Moreover, quite recently (\cite{GuVo}) a concise and complete classification of unitary irreducible representations of non-compact Lie superalgebras was presented, where the whole construction lies on bosonic and fermionic C/A's. 

Within the above context, our main goal   of constructing the NC relativistic space using oscillators is highly motivated. \\

Four-dimensional NC relativistic space (or relativistic fuzzy space) with conformal symmetry will be of our interest in this paper. There have not been many explicit realizations of this issue (\cite{Eck}):  the idea of Connes's almost non-commutative geometry (\cite{Con}), $\kappa$-deformed Minkowski space (\cite{KoLuMa}, \cite{CeWe}) or direct efforts of adding gravitational fields to NC spaces (\cite{Mad2}). All are classical and lack direct ability to be quantized. 

Moreover, the analysis of the unitary irreducible representations of the underlying group (or algebra) $SU(2,2)$ itself is quite a novel feature, too. The first original treatment of Mack (\cite{Ma}) proposed 5 classes of such representations, but the construction is far from explicit. Then \cite{GuVo} made this explicit and veryfied Mack's result using the oscillator construction. For the NC case, however, no such construction has been made so far. One would think that \cite{GuVo} is directly generalizable to the NC case, since it contains oscillators as well, but contrary is true: their construction uses oscillators, whose action is realised on a \textit{commutative} vector space $\mathcal{F}_\gamma$ viewed as a ring of functions (that posseses a real auxiliary parameter $\gamma$ giving rise to representations with non-integer highest weight). In our case, we will realize the non-commutativity of the space itself using oscillators forming a non-commutative algebra of functions. Thus our oscillators are purely "NC-structural" and have no connection to those in \cite{GuVo} (and neiter to the creation or annihilation of any particle as in \cite{We1}). 

In our paper, we want to reinterpret Mack's construction. Due to the differences mentioned above, our construction will be applicable only to two out of four non-trivial classes of representation classes, described only by half-integer $d$ (dimension), whereas the other two would need $d$ to be a general real number (cf. Section \ref{sec_discussion} for a discussion of this issue). This could be viewed as a consequence of the fuzziness of the space per analogy with the spectrum of quantum harmonic oscillator, that is discrete due to the non-commutativity of the C/A's. 

We proceed as follows: in the first part of the paper, we construct particular oscillatory representation of $so(2,4) \cong su(2,2)$ algebra. To do this, we construct non-commutative algebra of functions of set of C/A's, that act on a (fuzzy) Hilbert space. We then represent the $so(2,4)$ algebra as an action on such functions. To this end, we construct another (auxiliary) Hilbert space as a more practical replacement for this action. Choosing a convenient basis of states on this space (i. e. the coherent states, \cite{Pe}), we prove these states to be boost-invariant and from their properties we reveal that we actually constructed the doubleton irreducible representation, i. e. classical fuzzy massless field. Consequently the second part of the paper is dedicated to direct product of such representations yielding another class of representations, massive classical fuzzy fields. At the and, discussion with outlook can be found, where we outline further possible analysis of this kind with connection to a possible interplay between NC space and gravity within the context of Newman -- Penrose formalism (\cite{NP}).

\bigskip

\section{Oscillatory Doubleton Representation}
\subsection{Fundamental and dual representations of $su(2,2)$}

Let us briefly review the constuction of the generators of $su(2,2)$ from isomorphism of Lie algebras $su(2,2) \cong so(2,4)$. We use metrics $\eta_{ab} = (+----+)$ on the latter, $a,b = 0,1,\dots, 5$. Restriction onto subalgebra $so(1,3)$ (signature $(+---)$) comprises of indices $\mu, \nu = 0,1,2,3$.  

Manifestly covariant $so(2,4)$ algebra is given in terms of generators $S_{ab}$:
\begin{equation}
[S_{ab}, S_{cd}] = \eta_{ac} S_{bd} - \eta_{ad} S_{bc} - \eta_{bc} S_{ad} + \eta_{bd} S_{ac}.
\label{so24}
\end{equation} 
The Lorentz $so(1,3)$ (sub)algebra is generated by 10 elements $S_{\mu \nu}$ in a standard way. Moreover, $S_{\mu 4}, S_{\mu 5}$ are $so(1,3)$ four-vectors, hence a linear combination of linearly independent $\gamma_\mu, \gamma_5 \gamma_\mu$. This motivates the choice of the generators as follows \footnote{We use mathematical convention, i. e. commutators wihout $i$ on the r.h.s.: $[A_{mat},B_{mat}]=C_{mat}$. For the physical convention, use $A_{ph}= i A_{mat}$, which yields $[A_{ph},B_{ph}]=i\, C_{ph}$. \label{convention}}
\begin{subequations}
\begin{align}
S_{\mu \nu} = &\, - \frac{1}{4} [\gamma_\mu, \gamma_\nu], \\
S_{\mu 4} = &\, +\frac{i}{2} \, \gamma_5 \gamma_\mu, \\
S_{\mu 5} = &\, - \frac{i}{2}\, \gamma_\mu, \\
S_{45} =& \, +\frac{1}{2} \, \gamma_5.
\end{align}\label{generators}
\end{subequations}

A more convenient basis is choosen for the subalgebra $su(1,1)$ ($a,b = 4, 5$):
\begin{subequations}
\begin{align}
P_\mu = &\,c \left( S_{\mu 5} + S_{\mu 4} \right), \\
K_\mu = &\, c \left( S_{\mu 5} - S_{\mu 4} \right), \\
D = &\, S_{45},
\end{align}\label{gen_PKD}
\end{subequations}
i.e. momenta ($P_\mu$), special conformal transformations ($K_\mu$) and dilatation ($D$); $c \in \mathbb{C}$ is normalization constant, we choose $c=1$ henceforth. Then new commutator relations appear (cf. \cite{GuVo} for the physical convention): 
\begin{subequations}
\begin{align}
[S_{\mu \nu},D ] = &\, 0\ \\
[P_\mu,D] = &\, P_\mu, \\
[K_\mu,D] = &\, - K_\mu, \\
[P_\mu, P_\nu] = &\, 0 = [K_\mu, K_\nu], \\
[K_\mu, P_\nu] = &\, 2\, \left( S_{\mu \nu} - \eta_{\mu \nu} D \right).
\end{align}\label{CR_PKD}
\end{subequations}

For what follows, we need generators in fundamental (and dual) representation. We work in Dirac basis, where $\gamma$-matrices read
\begin{align}
\gamma^0 = \begin{pmatrix} 1 & 0 \\ 0 & -1 \end{pmatrix}, \, \,  \, \, \, 
\gamma^i = \begin{pmatrix} 0 & \sigma^i \\ - \sigma^i & 0 \end{pmatrix}, \, \, \, \, \, 
\gamma^5 = i \gamma^0 \, \gamma^1 \, \gamma^2 \, \gamma^3 = \begin{pmatrix} 0 & 1 \\ 1 & 0\end{pmatrix}, 
\end{align}
$\sigma^i$ being the Pauli matrices. Note that  $(\gamma_0, \gamma_i, \gamma_5) = (\gamma^0, - \gamma^i, - \gamma^5)$. 
Working out all generators explicitly yields their matrix form in the fundamental representation
\begin{subequations}
\begin{align}
S_{0i} =&\,  - \frac{1}{4}\,[\gamma_0, \gamma_i] = \frac{1}{4}\,[\gamma^0, \gamma^i]= \frac{1}{2}\,\begin{pmatrix} 0 & \sigma_i \\ \sigma_i & 0 \end{pmatrix}, \\
S_{ij}=&\, - \frac{1}{4}\, [\gamma_i , \gamma_j] = - \frac{1}{4}\, [\gamma^i , \gamma^j] = \frac{i}{2}\, \varepsilon_{ijk} \begin{pmatrix}
\sigma_k & 0 \\ 0 & \sigma_k \end{pmatrix},\\
K_0 = &\,\frac{i}{2}\, \begin{pmatrix} -1 & -1 \\ 1 & 1 \end{pmatrix}, 
\\
K_i = &\, \frac{i}{2} \, \begin{pmatrix}
\sigma_i & \sigma_i \\ - \sigma_i & - \sigma_i
\end{pmatrix}, \\
P_0 = &\,\frac{i}{2}\, \begin{pmatrix} -1 & 1 \\ -1 & 1 \end{pmatrix},
\\ 
P_i = &\, \frac{i}{2} \, \begin{pmatrix}
- \sigma_i & \sigma_i \\ - \sigma_i &  \sigma_i
\end{pmatrix}, \\
D = &\, - \frac{1}{2}\, \begin{pmatrix} 0 & 1 \\ 1 & 0 \end{pmatrix}.
\end{align}\label{gen_expl} 
\end{subequations}

\begin{remark}
For the matrix Lie algebra $\mathcal{G}$,  representation $\rho \colon \mathcal{G} \to gl(n, \mathbb{R})$ has its dual $\rho^\ast \colon \mathcal{G} \to gl(n, \mathbb{R})$ defined for each $X\in \mathcal{G}$ as $\rho^\ast(X) = - \rho(X)^T$. In our case, dual representation is unitary inequivalent to the fundamental one, because there is no such unitary matrix $\mathcal{U}$ that would ensure $\mathcal{U} \, \rho(X) \, \mathcal{U}^{-1} = \rho^\ast(X)$. We denote the dual generators by prime, $S'_{ab} = - S_{ab}^T$, or explicitely
\begin{subequations}
\begin{align}
S_{0i}^\prime =&\, (-1)^i \, S_{0i}, \\
S_{ij}^\prime=&\, \frac{i}{2}\, (-1)^k\, \varepsilon_{ijk} \begin{pmatrix}
\sigma_k & 0 \\ 0 & \sigma_k \end{pmatrix},\\
K_0^\prime = &\,- P_0,\\
K_i^\prime = &\, - (-1)^i\, P_i, \\
P_0^\prime = &\, -K_0, \\ 
P_i^\prime = &\, -(-1)^i\, K_i, \\
D^\prime = &\, - D
\end{align}\label{gen_expl_dual} 
\end{subequations}
as $(-\sigma_i)^T = (-1)^i \sigma_i$. 
\end{remark}

\bigskip

\subsection{Oscillatory representation(s)}
We are in position to construct the oscillatory representation. Before doing so, few definitions and conventions are needed. 

\begin{definition}[Fuzzy Hilbert space]
Let $a_\alpha$, $a^{\dagger \alpha}$, $\alpha=1,2$ be an auxiliary doublet of oscillators, i.e. creation and annihilation operators (C/A) with standard commutator relations (note the position of the indices)
\begin{align}
[a_\alpha, a^{\dagger \beta}]= \delta_\alpha^{\phantom \alpha \beta}, \, \, \, \, [a_\alpha, a_\beta] = 0 = [a^{\dagger \alpha}, a^{\dagger \beta}].
\end{align} 
Then (fuzzy)\footnote{We want to distinguish it from the ordinary Hilbert space of states, where $a^\dagger$ corresponds to creation of a particle state. Our $\mathcal{H}_F$ is purely an algebraic construction.} Hilbert space $\mathcal{H}_F$ of normalised states is constructed from vacuum state $|0\rangle \equiv |0,0 \rangle$, $\hat{a}_1 | 0 \rangle = 0 = \hat{a}_2 | 0 \rangle$, in a standard way
\begin{align}
\mathcal{H}_F := \{| n_1, n_2 \rangle \,, \, \,  n_1, n_2 \in \mathbb{N} \, \big\vert \, \, 
| n_1, n_2 \rangle = \frac{\left(a^{\dagger 1} \right)^{n_1} \left(a^{\dagger 2} \right)^{n_2} }{\sqrt{n_1! \, n_2!}} \, | 0, 0 \rangle\}. \label{fock_basis}
\end{align}
\end{definition}

\begin{remark}
\label{rm_fock_norm}
The norm on $\mathcal{H}_F$, $|| \, \vert n \rangle \, ||^2 := \langle n \vert n \rangle$, agrees with the normalization of the states (for fixed $\alpha$): 
\begin{equation}
a_\alpha \, \vert n_\alpha \rangle = \sqrt{n_\alpha}\, \vert n_\alpha-1 \rangle, \, \, \, \, \, a^{\dagger \alpha} \, \vert n_\alpha \rangle = \sqrt{n_\alpha+1}\, \vert n_\alpha + 1 \rangle. \label{norm_FF}
\end{equation}
\end{remark}

This construction enables us to define one of the crutial objects of our investigation.

\begin{definition}[Fuzzy function]
Let $\mathcal{A}:= \End(\mathcal{H}_F)$. Then $\Psi(a^\dagger,a) \in \mathcal{A}$ is a fuzzy function. It is an analytic function of C/A's written in normal order\footnote{We use the standard notation: if $\chi$ is an arbitrary product of C/A's, then $\colon \chi \colon$ denotes its normal order.\label{normal} }, i.e. with all creation operators to the left of all annihilation operators to the right, 
\begin{align}
\Psi(a^\dagger, a) = \sum_{n_1, n_2, m_1, m_2} c_{n_1 n_2 m_1 m_2} (a^{\dagger 1})^{n_1} (a^{\dagger 2})^{n_2} (a_1)^{m_1} (a_2)^{m_2}. \label{expansion}
\end{align}
We use Hilbert --  Schimdt norm $|| \Psi ||^2_{\mathcal{A}} = \langle \Psi, \Psi \rangle = \Tr \left( \Psi^\dagger \Psi \right)$ on $\mathcal{A}$.  
\end{definition}

\begin{remark}[Chiral parameter]
Fuzzy function acquires a phase factor under chiral tranformation $a_\alpha \mapsto e^{i\, \theta} a_\alpha$:
\begin{align}
\Psi \left(e^{- i \alpha} a^\dagger, e^{i \alpha} a \right)= e^{- i \kappa \alpha} \, \Psi(a^\dagger, a). \label{chirality}
\end{align}
Here $\kappa = (n_1 + n_2) - (m_1 + m_2) \in \mathbb{Z}$ governs the difference between the number of $a^\dagger$ and $a$ in the power expansion of $\Psi$ and we will denote it the chiral parameter. 
\end{remark}

Our represented $su(2,2)$ generators will act on such functions. To construct them, we need to define operators on $\mathcal{A}$\footnote{We go with the convention, that operators on $\mathcal{A}$ are hatted.}. 

\begin{definition}[Auxiliary oscillators]
Let $\Psi \in \mathcal{A}$. Then auxiliary set $\hat{a}_\alpha, \hat{b}^{\dot{\alpha}} \in \End(\mathcal{A})$ with their conjugates acts on fuzzy function as follows:\begin{subequations}
\begin{align}
\hat{a}_\alpha  \Psi :=&\, a_\alpha  \Psi, \\
\hat{a}^{\dagger \alpha}  \Psi :=&\, a^{\dagger \alpha}  \Psi, \\
\hat{b}^{\dot{\alpha}}  \Psi := &\, \Psi  \varsigma^{\dot{\alpha} \alpha} a_\alpha , \\
\hat{b}_{\dot{\alpha}}^\dagger \Psi := &\, \Psi a^{\dagger \alpha} \varsigma_{\alpha \dot{\alpha}},
\end{align}
\end{subequations}
where $\varsigma_{1 \dot{1}} = \varsigma_{2 \dot{2}} = \varsigma^{\dot{1} 1 } = \varsigma^{\dot{2} 2} = 1$, $\varsigma_{1 \dot{2}} = \varsigma_{2 \dot{1}} = \varsigma^{\dot{1} 2 } = \varsigma^{\dot{2} 1} = 0$ and with the summation convention $\phantom 1^\alpha \searrow \phantom 1_\alpha, \phantom 1_{\dot{\alpha}} \nearrow \phantom 1^{\dot{\alpha}}$, cf. \cite{DrHa}.  
\end{definition}

\begin{remark}
From the construction obviously $[\hat{a}_\alpha, \hat{a}^{\dagger \beta}] =\, \delta_\alpha^{\phantom \alpha \beta}, \, \, \, \, [\hat{b}^{\dot{\alpha}}, \hat{b}_{\dot{\beta}}^\dagger] =\, - \delta^{\dot{\alpha}}_{\phantom \alpha \dot{\beta}}$, and all other commutators are zero. 
\end{remark}

\begin{remark}
We do not associate $\varsigma$'s (with bisponorial index structure) neither with any four-vector\footnote{Cf. (2.32) in \cite{DrHa}, $V_{a \dot{\beta}} \leftrightarrow V^\mu = \frac{1}{2} \bar{\sigma}^{\mu \dot{\beta} \alpha} V_{\alpha \dot{\beta}}$.}, nor with symplectic form on $sl(2,\mathbb{C})$ ( i.e. Levi-Civita symbol $\varepsilon$ used for lowering / raising indices in proper spinorial context). It is only Kronecker $\delta$ symbol with mixed dotted/undotted indices refering to the off-diagonal block position, i.e. $\varsigma_{\alpha \dot{\alpha}} \varsigma^{\dot{\alpha} \beta} = \delta_\alpha^{\phantom \alpha \beta}, \, \, \,  \varsigma^{\dot{\alpha} \alpha} \varsigma_{\alpha \dot{\beta}} = \delta^{\dot{\alpha}}_{\phantom \alpha \dot{\beta}}$. 
\end{remark}

\begin{definition}[Auxiliary (modified) Hilbert space]
Let $\hat{a}_\alpha, \, \hat{a}^{\dagger \alpha}, \, \hat{b}^{\dot{\alpha}} ,\,  \hat{b}^\dagger_{\dot{\alpha}} \in \End(\mathcal{A})$ be as above. Then a general element $h = \vert n_1, n_2 \rangle \langle m_1, m_2 \vert \equiv \vert n_1, n_2 \rangle \otimes \id_{\mathcal{A}} \otimes \langle m_1, m_2 \vert$ of the auxiliary Hilbert space $\mathcal{H}_A \simeq \mathcal{H}_F \otimes \mathcal{A} \otimes \overline{\mathcal{H}}_F$ is defined by its action on function $\Psi \in \mathcal{A}$ (viewed as $\mathcal{A} \hookrightarrow \mathcal{H}_A$, i.e. $\Psi \equiv \id_{\mathcal{H}_F} \otimes \Psi \otimes \id_{\overline{\mathcal{H}}_F}$) as follows: 
\begin{equation}
h \circ \Psi \equiv \vert n_1, n_2 \rangle \otimes \Psi \otimes \langle m_1, m_2 \vert := (a^{\dagger 1 })^{n_1}\, (a^{\dagger 2})^{n_2} \, \Psi \, (a_1)^{m_1} \, (a_2)^{m_2}. 
\end{equation}
\end{definition}

\begin{remark}[Notation]
One should distinguish the basis elements on $\mathcal{H}_F$ from the alike looking (part of) basis elements on $\mathcal{H}_A$, as they have different normalisation and different action of C/A's on them. For the sake of notational simplicity, however, we will abuse the notation and use the ordinary bra-ket notion in both cases. It should be clear from the context, whether we deal with states on $\mathcal{H}_F$ or $\mathcal{H}_A$. 
\end{remark}

\begin{remark}
\label{monomial}
The fuzzy monomials $(a^{\dagger 1})^{n_1} \, (a^{\dagger 2})^{n_2} \,  (a_1)^{m_1} \, (a_2)^{m_2}$ can be viewed as an action of $\vert n_1, n_2 \rangle \langle m_1, m_2 \vert \equiv \vert n_1, n_2 \rangle \otimes \id_{\mathcal{A}} \otimes \langle m_1, m_2 \vert$ on $1_{\mathcal{H}_A} \equiv \id_{\mathcal{H}_F} \otimes 1_\mathcal{A} \otimes \id_{\overline{\mathcal{H}}_F}$ . Thus action of any $f \in \mathcal{H}_A$ on any general $\Psi \in \mathcal{A} \hookrightarrow \mathcal{H}_A$ can be reduced to the action on $1_{\mathcal{H}_A}$. 
\end{remark}

Naively, one would construct Hilbert spaces similar to $\mathcal{H}_F$ for both C/A sets ($a$'s and $b$'s), but it would not work. Let us take $\psi = (a^{\dagger 1})^{n_1} \, (a^{\dagger 2})^{n_2} \,  (a_1)^{m_1} \, (a_2)^{m_2}$. Acting with $\hat{a}^{\dagger \alpha}$ or $\hat{b}^{\dot{\alpha}}$ meets the normal ordering, which is, however, not true for the remaining two cases. The modification is summarized in the following

\begin{prop}
\label{prop_fockA}
Action of $\hat{a}_\alpha, \, \hat{a}^{\dagger \alpha}, \, \hat{b}^{\dot{\alpha}} ,\,  \hat{b}^\dagger_{\dot{\alpha}} \in \End(\mathcal{A})$ on the auxiliary Hilbert space $\mathcal{H}_A = \spa \{ \vert n_1, n_2 \rangle \langle m_1, m_2 \vert \}$ reads (for fixed $\alpha$):
\begin{subequations}
\begin{align}
\hat{a}^{\dagger \alpha} \, \vert  n_\alpha \rangle \langle  m_\alpha \vert  =&\, \vert  n_\alpha+1  \rangle \langle  m_\alpha \vert, \\
\hat{a}_\alpha \, \vert  n_\alpha \rangle \langle m_\alpha \vert =& \, \vert  n_\alpha \rangle \langle  m_\alpha+1  \vert + n_\alpha \, \vert  n_\alpha-1  \rangle \langle  m_\alpha \vert,\\
\hat{b}^{\dot{\alpha}} \, \vert  n_\alpha \rangle \langle m_\alpha \vert =& \, \vert  n_\alpha \rangle \langle m_\alpha+1 \vert, \\
\hat{b}^\dagger_{\dot{\alpha}}\,  \vert  n_\alpha \rangle \langle m_\alpha \vert =&  \vert  n_\alpha+1 \rangle \langle m_\alpha \vert + m_\alpha \vert  n_\alpha \rangle \langle m_\alpha -1\vert.
\end{align}
\end{subequations}
\end{prop}

\begin{proof}
Cases a), c) are trivial and could be understood as a \textit{normalization condition}. As for b), d), one has to restore the normal ordering. An elementary calculation using
\begin{equation}
[a_\alpha, (a^{\dagger \alpha})^n] = n\, (a^{\dagger \alpha})^{n-1} \implies [a_\alpha, \cdot\, ] = \frac{\partial}{\partial a^{\dagger \alpha}}\, \, \, \, \, \mbox{and}\, \, \, \, \, [a^{\dagger \alpha}, (a_\alpha)^n] = -n\, (a_\alpha)^{n-1} \implies [a^{\dagger \alpha}, \cdot \,] = - \frac{\partial}{\partial a_\alpha}
\label{a_der}
\end{equation}  
shows, that 
\begin{equation}
\hat{a}_\alpha = \varsigma_{\alpha \dot{\alpha}} \, \hat{b}^{\dot{\alpha}} + \frac{\partial}{\partial a^{\dagger \alpha}} \, \, \, \, \, \mbox{and} \, \, \, \, \,  \hat{b}^\dagger_{\dot{\alpha}} = \hat{a}^{\dagger \alpha} \, \varsigma_{\alpha \dot{\alpha}} + \frac{\partial}{\partial a_\alpha} \, \varsigma_{\alpha \dot{\alpha}} \, \, \, \, \,  \mbox{on} \, \, \,  \mathcal{A}.
\end{equation}
Thus action of $\hat{a}_\alpha, \, \hat{b}^\dagger_{\dot{\alpha}}$ can be recasted into action of $\hat{a}^{\dagger \alpha}, \, \hat{b}^{\dot{\alpha}}$, that respects the normal ordering and is natural on $\mathcal{H}_A$. 
\end{proof}

\begin{remark}
This result can be also checked directly on $\vert \mu_1, \mu_2 \rangle \in \mathcal{H}_F$. Take fuzzy monomial $\psi = (a^{\dagger 1})^{n_1} \, (a^{\dagger 2})^{n_2} \,  (a_1)^{m_1} \, (a_2)^{m_2}$, then $\psi \, \vert \mu_1, \mu_2 \rangle = \prod_{\alpha} \sqrt{\frac{(\mu_\alpha-m_\alpha+n_\alpha)! \, \mu_\alpha!}{(\mu_\alpha-m_\alpha)!^2}} \, \vert \mu_1-m_1+n_1, \mu_2-m_2+n_2 \rangle$ (cf. Remark \ref{rm_fock_norm} for the normalization). In the same fashion one can easily show, that e.g. $a_1\,  \psi \, \vert \mu_1, \mu_2 \rangle = \left( (a^{\dagger 1})^{n_1} \, (a^{\dagger 2})^{n_2} \,  (a_1)^{m_1+1} \, (a_2)^{m_2} 
+n_1\, (a^{\dagger 1})^{n_1-1} \, (a^{\dagger 2})^{n_2} \,  (a_1)^{m_1} \, (a_2)^{m_2} \right) \vert \mu_1, \mu_2 \rangle$, the rest analogously. 
\end{remark}

\begin{remark}
Considering the HS norm, one has to be cautious about a minor modification due to another normalization of states on $\mathcal{H}_A$: 
\begin{equation}
||\Psi ||^2_{\mathcal{H}_A} = \Tr ( \Psi^\dagger\, \Psi) = \sum_k \frac{1}{k_1!\, k_2!}\, \langle k_1, k_2 \vert \Psi^\dagger\, \Psi \vert k_1, k_2 \rangle.
\end{equation}
\end{remark}

Now is everything ready for the construction of the oscillatory representation of $su(2,2)$ algebra. 

\begin{prop}[Oscillatory representation]
\label{prop_oscil}
Let $\hat{A}_{\underline{\alpha}} := \left( \hat{a}_\alpha \, \, \hat{b}^{\dot{\alpha}} \right)^T$ and $\tilde{A}_{\underline{\alpha}} :=  \left(\tilde{a}_\alpha \, \, \tilde{b}^{\dot{\alpha}} \right)^T$, where $\underline{\alpha}$ is multi-index for both dotted and undotted indices copyring the position of the undotted one, moreover let $\Gamma_{\underline{\alpha}}^{\phantom \alpha \underline{\beta}} := 
[\hat{A}_{\underline{\alpha}}, \hat{A}^{\dagger \underline{\beta}} ] \equiv [\tilde{A}_{\underline{\alpha}}, \tilde{A}^{\dagger \underline{\beta}} ]$ (cf. Remark \ref{notation_hattilde} for the notation).
Then
\begin{subequations}
\begin{align}
\hat{S}_{ab} = &\, \hat{A}^\dagger \, \Gamma \, S_{ab} \, \hat{A}, \\
\tilde{S}_{ab} = &\, \tilde{A}^\dagger \, \Gamma \, S_{ab}^\prime \, \tilde{A} 
\end{align}\label{osc_rep}
\end{subequations}
is the oscillatory unitary fundamental and dual representation respectively. 
\end{prop}

\begin{proof}
We will prove it for the fundamental representation, the dual one is analogous. 

\begin{align}
[\hat{S}_{ab}, \hat{S}_{cd}] = &\, (\hat{A}^\dagger \Gamma)^{\underline{\alpha}} \, (S_{ab})_{\underline{\alpha}}^{\phantom \alpha \underline{\beta}} \hat{A}_{\underline{\beta}} \, (\hat{A}^\dagger \Gamma)^{\underline{\gamma}} \, (S_{cd})_{\underline{\gamma}}^{\phantom \gamma \underline{\delta}} \hat{A}_{\underline{\delta}} - (\hat{A}^\dagger \Gamma)^{\underline{\gamma}} \, (S_{cd})_{\underline{\gamma}}^{\phantom \gamma \underline{\delta}} \hat{A}_{\underline{\delta}}\, (\hat{A}^\dagger \Gamma)^{\underline{\alpha}} \, (S_{ab})_{\underline{\alpha}}^{\phantom \alpha \underline{\beta}} \hat{A}_{\underline{\beta}} \notag \\
= &\, (S_{ab})_{\underline{\alpha}}^{\phantom \alpha \underline{\beta}} \, (S_{cd})_{\underline{\gamma}}^{\phantom \gamma \underline{\delta}} \left[(\hat{A}^\dagger \Gamma)^{\underline{\alpha}} \hat{A}_{\underline{\beta}} \, (\hat{A}^\dagger \Gamma)^{\underline{\gamma}} \hat{A}_{\underline{\delta}} - (\hat{A}^\dagger \Gamma)^{\underline{\gamma}} \hat{A}_{\underline{\delta}} \, (\hat{A}^\dagger \Gamma)^{\underline{\alpha}} \hat{A}_{\underline{\beta}}\right] \notag \\
=&\, (S_{ab})_{\underline{\alpha}}^{\phantom \alpha \underline{\beta}} \, (S_{cd})_{\underline{\gamma}}^{\phantom \gamma \underline{\delta}} \left[(\hat{A}^\dagger \Gamma)^{\underline{\alpha}} \left(\delta_{\underline{\beta}}^{\phantom \beta \underline{\gamma}} - (\hat{A}^\dagger \Gamma)^{\underline{\gamma}} \hat{A}_{\underline{\beta}} \right) \hat{A}_{\underline{\delta}} - (\hat{A}^\dagger \Gamma)^{\underline{\gamma}} \left(\delta_{\underline{\delta}}^{\phantom \beta \underline{\alpha}} - (\hat{A}^\dagger \Gamma)^{\underline{\alpha}} \hat{A}_{\underline{\delta}} \right) \hat{A}_{\underline{\beta}} \right] \notag \\
= &\, (S_{ab})_{\underline{\alpha}}^{\phantom \alpha \underline{\beta}} \, (S_{cd})_{\underline{\gamma}}^{\phantom \gamma \underline{\delta}} \left[(\hat{A}^\dagger \Gamma)^{\underline{\alpha}}\, \delta_{\underline{\beta}}^{\phantom \beta \underline{\gamma}} \hat{A}_{\underline{\delta}} - (\hat{A}^\dagger \Gamma)^{\underline{\gamma}}\, \delta_{\underline{\delta}}^{\phantom \beta \underline{\alpha}}  \hat{A}_{\underline{\beta}} \right] \notag \\
= &\,(\hat{A}^\dagger \Gamma)^{\underline{\alpha}} \left[  (S_{ab})_{\underline{\alpha}}^{\phantom \alpha \underline{\beta}} \, (S_{cd})_{\underline{\beta}}^{\phantom \beta \underline{\gamma}}  - (S_{cd})_{\underline{\alpha}}^{\phantom \alpha \underline{\beta}} \, (S_{ab})_{\underline{\beta}}^{\phantom \beta \underline{\gamma}}  \right] \hat{A}_{\underline{\gamma}} \notag \\
= &\, \hat{A}^\dagger \Gamma [S_{ab}, S_{cd}] \hat{A} = \widehat{[S_{ab}, S_{cd}]}. 
\end{align}

Representation $\rho$ of group $G$ on Hilbert space $\mathcal{H}$ is unitary, if $\mathcal{U} = \rho(g)$ is a unitary operator for any $g\in G$. Operator $\mathcal{U}$ on $\mathcal{H}$ is unitary, when it is surjecive and preserves the scalar product, i.e. $\langle \mathcal{U} \Psi, \mathcal{U} \Phi \rangle = \langle \Psi, \Phi \rangle$ for any $\Psi, \Phi \in \mathcal{H}$ (recall $\langle \Psi, \Phi \rangle_{HS} = \Tr ( \Psi^\dagger \Phi)$, thus we need $\mathcal{U}^\dagger \mathcal{U} = \mathcal{U} \mathcal{U}^\dagger = 1$). In our case, $g = e^S$ for $S$ being a generator of $su(2,2)$. From the matrix $SU(2,2)$ definition ($\Gamma$ is metrics on $SU(2,2)$) we have the following condition on $su(2,2)$ generators
\begin{align}
g^\dagger \Gamma g = \Gamma, \, \, \, \Gamma = \begin{pmatrix}
1 & 0 \\ 0 & - 1\end{pmatrix} \implies S^\dagger \Gamma + \Gamma S = 0.
\end{align}

It can be easily solved for $S$ in terms of $2 \times 2$ blocks (see bellow) as follows:
\begin{align}
(S_\alpha^{\phantom \alpha \beta} )^\dagger = - S_\alpha^{\phantom \alpha \beta}, \, \, \, (S^{\dot{\alpha}}_{\phantom \alpha \dot{\beta}})^\dagger = - S^{\dot{\alpha}}_{\phantom \alpha \dot{\beta}}, \, \, \, (S_{\alpha \dot{\beta}})^\dagger = S^{\dot{\alpha} \beta}, \, \, \, (S^{\dot{\alpha} \beta})^\dagger = S_{\alpha \dot{\beta}}
\end{align} 
i.e. diagonal blocks are skew-hermitian and off-diagonal blocks are mutually hermitian conjugated. Hence
\begin{align}
\hat{S}^\dagger = &\, (\hat{A}^\dagger \, \Gamma \, S \, \hat{A})^\dagger =\hat{A}^\dagger \, S^\dagger \, \Gamma \, \hat{A} = \hat{a}^\dagger S^\dagger \hat{a} + \hat{b}^\dagger S^\dagger \hat{a} - \hat{a}^\dagger S^\dagger \hat{b} + \hat{b}^\dagger S^\dagger \hat{b} \notag \\
=&\, - \left( \hat{a}^\dagger S \hat{a} - \hat{b}^\dagger S \hat{a} + \hat{a}^\dagger S \hat{b} - \hat{b}^\dagger S \hat{b} \right) = - \hat{S}. 
\end{align}
Then if $\rho(S) = \hat{S}$ and $g=e^S$,  $\mathcal{U} = e^{\hat{S}}$ and we have $\mathcal{U}^\dagger = e^{\hat{S}^\dagger} = e^{-\hat{S}}$, hence $\mathcal{U}^\dagger \mathcal{U} = \mathcal{U} \mathcal{U}^\dagger = 1$. 
\end{proof}

\begin{remark}[Block structure]
The block structure of the generators is
\begin{align}
S = \begin{pmatrix}
S_{\alpha}^{\phantom \alpha \beta} & S_{\alpha \dot{\beta}} \\
S^{\dot{\alpha} \beta} & S^{\dot{\alpha}}_{\phantom \alpha \dot{\beta}} 
\end{pmatrix}. \label{blocks}
\end{align}
Then e.g.
\begin{align}
\hat{S} = \hat{A}^\dagger \, \Gamma \, S \, \hat{A} \equiv \hat{A}^{\dagger \underline{\alpha}} \, \Gamma_{\underline{\alpha}}^{\phantom \alpha \underline{\beta}} \, S_{\underline{\beta}}^{\phantom \beta \underline{\gamma}} \, \hat{A}_{\underline{\gamma}} = &\, \hat{a}^{\dagger \alpha} S_\alpha^{\phantom \alpha \beta} \hat{a}_\beta - \hat{b}^\dagger_{\dot{\alpha}} S^{\dot{\alpha} \beta} \hat{a}_\beta + \hat{a}^{\dagger \alpha} S_{\alpha \dot{\beta}} \hat{b}^{\dot{\beta}} - \hat{b}^\dagger_{\dot{\alpha}} S^{\dot{\alpha}}_{\phantom \alpha \dot{\beta}} \hat{b}^{\dot{\beta}} \notag \\
= &\, \hat{a}^\dagger S \hat{a} - \hat{b}^\dagger S \hat{a} + \hat{a}^\dagger S \hat{b} - \hat{b}^\dagger S \hat{b},
\end{align}
where we abused the notation in the second line by dropping the indices, because $2 \times 2$ blocks of $S$ are uniquely determined by $\hat{a}$'s and $\hat{b}$'s, whose indices are on fixed positions. 
\end{remark}

\begin{remark}
One may wish to consider a more general matrix in the construction of the representation, i.e.  $\hat{S}_{ab} = \hat{A}^\dagger \, R \, S_{ab} \, \hat{A}$. Then, however, the condition on $\hat{S}_{ab}$ to be a representation together with the fact that$\Gamma_{\underline{\alpha}}^{\phantom \alpha \underline{\beta}} := [\hat{A}_{\underline{\alpha}}, \hat{A}^{\dagger \underline{\beta}} ] $ yields $R = \Gamma^{-1}$, which in our case gives $\Gamma^{-1} = \Gamma$. Thus our construction is unique (within this context). 
\end{remark}

At this stage, it is quite instructive to point out the connection to general non-commutative (NC) relation among coordinates of a fuzzy space. Recall, that fuzzy space can be given in terms of fuzzy coordinates $x_i$ whose non-commutativity is governed by $\Theta_{ij}(x)$ in a standard way (cf. ftnt \ref{convention} for convention), $[x_i, x_j] = \lambda \, \Theta_{ij}(x)$, where $\lambda$ is the NC scale (and $\lambda \to 0$ is the commutative limit) and $\Theta_{ij}$ contains information about the structure constants of the underlying algebra. For instance, the Schwinger construction for $su(2)$  (i.e. $S^2_\lambda$ space) has $x^i = \lambda\, a^\dagger\,  \sigma^i \, a$, which yields $\Theta_{ij}(x) = 2 \epsilon_{ijk} x^k$, where $f_{ijk} = 2\, \epsilon_{ijk}$ are the corresponding structure constants.  In a similar way, if we take $x_{ab} = \lambda \, \hat{S}_{ab}$, we immediately get (cf. Proposition \ref{prop_oscil}) $[x_{ab}, x_{cd}] = \lambda^2 \, \hat{A}^\dagger \Gamma [S_{ab}, S_{cd}] \hat{A} = \lambda^2 \hat{A}^\dagger \Gamma {f_{abcd}}^{ef}\, S_{ef} \hat{A} = \lambda \, {f_{abcd}}^{ef} \, x_{ef}$, thus $\Theta_{abcd}(x) = {f_{abcd}}^{ef} x_{ef}$. Hence we can view $\hat{S}_{ab}$ as NC coordinates of the fuzzy $SU(2,2)_\lambda$ space. 
\begin{remark}
If one seeks for the commutative limit, then one has to replace $a_\alpha \mapsto \sqrt{\lambda}\, a_\alpha$ on $\mathcal{H}_F$ and then take $\lambda \to 0$ in the results. 
\end{remark}

Eventually, let us write explicitly the generators in oscillatory fundamental representation for future reference:
\begin{subequations}
\begin{align}
\hat{S}_{0i} = &\, \frac{1}{2}\, \left(\hat{a}^\dagger \sigma_i \hat{b} - \hat{b}^\dagger \sigma_i \hat{a} \right), \\
\hat{S}_{ij} = &\, \frac{i}{2} \varepsilon_{ijk} \left( \hat{a}^\dagger \sigma_k \hat{a} - \hat{b}^\dagger \sigma_k \hat{b} \right), \\
\hat{K}_0 = &\, - \frac{i}{2} \left(\hat{a}^\dagger + \hat{b}^\dagger \right)\left( \hat{a} + \hat{b} \right), \\
\hat{K}_i = &\, \frac{i}{2} \left(\hat{a}^\dagger + \hat{b}^\dagger \right)  \sigma_i  \left( \hat{a} + \hat{b} \right), \\
\hat{P}_0 = &\, - \frac{i}{2} \left(\hat{a}^\dagger - \hat{b}^\dagger \right)\left(\hat{a} - \hat{b} \right), \\
\hat{P}_i = &\, -\frac{i}{2} \left(\hat{a}^\dagger - \hat{b}^\dagger \right) \sigma_i \left(\hat{a} - \hat{b} \right), \\
\hat{D} = &\, \frac{1}{2} \left(\hat{b}^\dagger \hat{a} - \hat{a}^\dagger \hat{b} \right), \\
\hat{C}_1 = &\, \frac{1}{2} \left( \hat{a}^\dagger \hat{a} - \hat{b}^\dagger \hat{b} \right), 
\end{align}\label{gen_osc}
\end{subequations}
where we added $\hat{C}_1 :=  \widehat{\frac{1}{2}\, 1_{4 \times 4}}$, the central extension element. For the dual representation, see (\ref{gen_expl_dual}) for corresponding changes.  

\bigskip

\subsection{Norm of the boosted state on $\mathcal{H}_A$}
Let us denote $\Psi(\beta)$ the boosted basis state on $\mathcal{H}_A$, $\Psi(\beta) := e^{\beta \, \hat{n}_i \, \hat{S}_{0i}} \, \Psi_{nm} \equiv e^{\beta\,  \hat{n}_i \, \hat{S}_{0i}} \, \vert n_1, n_2 \rangle \langle m_1, m_2 \vert$ for some unit vector $\hat{n}_i$. We choose $\hat{n}_i =(0,0,1)$, i.e.  the boost in $3$-direction. We want to investigate, whether is the norm on $\mathcal{H}_F$ boost-invariant, i.e. $|| e^{\beta \hat{S}_{03}} \, \Psi \, \vert \mu_1, \mu_2 \rangle \, ||^2 \overset{\underset{\mathrm{?}}{}}{=} || \Psi \, \vert \mu_1, \mu_2 \rangle \, ||^2$ for any $\Psi \in \mathcal{A}$ and $| \mu_1, \mu_2 \rangle \in \mathcal{H}_F$. Recall that any $\Psi$ can be constructed from fuzzy monomials (cf. Remark \ref{monomial}) that can be viewed as $\vert n_1, n_2 \rangle \langle m_1, m_2 \vert \circ 1_{\mathcal{H}_A}$. Thus it is sufficient to determine whether $\big \vert \big \vert e^{\beta \hat{S}_{03}} \, \vert n_1, n_2 \rangle \langle m_1, m_2 \vert \big \vert \big \vert^2 \overset{\underset{\mathrm{?}}{}}{=} \big \vert \big \vert  \vert n_1, n_2 \rangle \langle m_1, m_2 \vert \big \vert \big \vert^2$, i.e. it is desirable to analyse the norm of the boosted state on $\mathcal{H}_A$ instead of $\mathcal{H}_F$. However, current basis $ \vert n_1, n_2 \rangle \langle m_1, m_2 \vert$ appears to be inconvenient, thus we will work in a more appropriate basis, the coherent states (CS) (\cite{Pe}). We will then be able to reconstruct the result in the former basis from the result in CS basis, see Appendix \ref{sec_ap_A}.

A (overcomplete) basis of coherent states on a Hilbert space generated from vacuum $\vert 0 \rangle$ by action of $a^\dagger$ is constructed for $\alpha \in \mathbb{C}$ as follows (\cite{Pe})
\begin{equation}
\vert \alpha \rangle := e^{-|\alpha|^2 /2 } \, e^{\alpha a^\dagger} \, \vert 0 \rangle
\label{coherent}
\end{equation}
and depending on the normalization of the states, one can expand the exponential.  Crucial property of CS is $\hat{a} \, \vert \alpha \rangle = \alpha \, \vert \alpha \rangle$ and among others are: the dot product  $\langle \alpha | \beta \rangle = e^{- | \alpha |^2/2 - | \beta |^2/2 + \bar{\alpha} \beta}$, hence $| \langle \alpha | \beta \rangle |^2 = e^{- | \alpha- \beta |^2} $. Finally, invariant norm is $d \alpha = \frac{1}{\pi} \, d \alpha^R \, d \alpha^I$, where $\alpha = \alpha^R + i \alpha^I$,  thus $\hat{1} = \frac{1}{\pi} \int d \alpha^R d \alpha^I | \alpha \rangle \langle \alpha |$.

In our case, we will construct CS on $\mathcal{H}_A$, where  effectively $\hat{a}^\dagger \, \vert n \rangle  = \vert n+1 \rangle$ (cf. Proposition \ref{prop_fockA}), hence $\langle m \vert n \rangle = n! \, \delta_{nm}$ and thus $\vert \alpha \rangle = e^{-|A|^2/2}\, \sum_{n=0}^\infty \, \frac{\alpha^n}{n!} \, \vert n \rangle$, which is altered in comparison with the original CS states as constructed in \cite{Pe}. Nevertheless, we will call them CS states by the abuse of the definition (\ref{coherent}) and we will use capital letters for CS (as the greek alphabet is already used for indices). Thus the starting point of our analysis is the CS basis on $\mathcal{H}_A$: 
\begin{equation}
\Phi_{A, B}:= | A_1, A_2 \rangle \langle B_1, B_2 \vert = e^{- ( | A_1 |^2 + | A_2 | + | B_1 |^2 + |B_2 |^2)/2} \, \sum_{n_1, n_2, m_1,m_2=0}^\infty \frac{A_1^{n_1}\,  A_2^{n_2} \, \bar{B}_1^{m_1}\,  \bar{B}_2^{m_2}}{n_1!\, n_2!\, m_1!\, m_2!}\, |n_1, n_2 \rangle \langle m_1, m_2 \vert .
\end{equation}

Our task is then  to compute the norm of $\Phi(\beta) = e^{\beta \hat{S}_{03}} \, \Phi_{A, B}$. We insert $\hat{1} = \int d C_1 \, d C_2 \, |C_1, C_2 \rangle \langle C_1, C_2| $: 
\begin{align}
|| \Phi(\beta) ||^2_{\mathcal{H}_A} =&\, \Tr[ \Phi(\beta)^\dagger \, \Phi(\beta)] \equiv \int d D_1 \, dD_2 \, \langle D_1, D_2 | \Phi(\beta)^\dagger \, \Phi(\beta) | D_1, D_2  \rangle \notag \\
=&\, \int d C_1 \, d C_2 \, d D_1 \, dD_2 \, \langle D_1, D_2 | \Phi(\beta)^\dagger |C_1, C_2 \rangle \langle C_1, C_2| \Phi(\beta) | D_1, D_2  \rangle \notag \\
=&\,  \int d C_1 \,d C_2 \, d D_1 \, dD_2 \, | \langle C_1, C_2| \Phi(\beta) | D_1, D_2  \rangle |^2 .\label{HS_norm}
\end{align}
Thus we focus on the matrix element $\langle \Phi(\beta) \rangle_{C D} =  \langle C_1, C_2| \, \Phi(\beta)\,  | D_1, D_2  \rangle$. We will proceed in the following steps:
\begin{enumerate}[label=(\roman*)]
\item Identification of $sl(2,\mathbb{C})$ algebra among terms in $\hat{S}_{03}$ and construction of an appropriate algebra ismomorphism for a more convenient action on $\mathcal{H}_A$ .
\item Gauss (Cartan) decomposition for the boost.
\item Succesive action of decomposed boost on CS. 
\end{enumerate}

Recall that $\hat{S}_{03} \Psi \equiv \hat{S}_{03}^{\bar{1}} \Psi -\hat{S}_{03}^{\bar{2}} \Psi$ with $\hat{S}_{03}^{\bar{\alpha}} \Psi = \frac{1}{2} \left( a^{\dagger \alpha} \Psi a_\alpha - a_\alpha \Psi a^{\dagger \alpha} \right)$ for fixed $\alpha$. We divide this operator into NC parts, $\hat{S}_{03}^{\bar{\alpha}} = \hat{T}_+^{\bar{\alpha}} + \hat{T}_-^{\bar{\alpha}}$, with (for fixed $\alpha$)
\begin{equation}
\hat{T}_+^{\bar{\alpha}} =\, +\frac{1}{2}\, \hat{a}^{\dagger \alpha} \,\varsigma_{\alpha \dot{\alpha}} \, \hat{b}^{\dot{\alpha}} ,\, \, \, \, \, 
 \hat{T}_-^{\bar{\alpha}} =\, - \frac{1}{2}\, \hat{b}_{\dot{\alpha}}^\dagger\, \varsigma^{\dot{\alpha} \alpha} \, \hat{a}_\alpha, \, \, \, \, \, 
 \hat{T}_0^{\bar{\alpha}} = \, \frac{1}{2} \left( \hat{a}^{\dagger \bar{\alpha}}\, \hat{a}_\alpha + \hat{b}_{\dot{\alpha}}^\dagger \,\varsigma^{\dot{\alpha} \alpha} \, \varsigma_{\alpha \dot{\beta}} \, \hat{b}^{\dot{\beta}} \right).
\end{equation}
These operators satisfy the $sl(2,\mathbb{C})$ algebra for both $\alpha = 1, 2$:
\begin{align}
[\hat{T}_+^{\bar{\alpha}}, \hat{T}_-^{\bar{\alpha}}] = \frac{1}{2} \,  \hat{T}_0^{\bar{\alpha}} , \, \, \, [ \hat{T}_0^{\bar{\alpha}},  \hat{T}_\pm^{\bar{\alpha}} ] = \pm \, \hat{T}_\pm^{\bar{\alpha}}.
\label{sl2_T}
\end{align}
In what follows, we drop the $\alpha$ index and carry out the computation for fixed $\alpha$. We then reconstruct the full result provided $[\hat{S}_{03}^{\bar{1}}, \hat{S}_{03}^{\bar{2}}]=0$. For this purpose, denote  $\hat{S} = \hat{T}_+ + \hat{T}_-$, $ \phi_{A,B} = \vert A \rangle \langle B \vert$ and $\phi(\beta) = e^{\beta \hat{S}} \, \phi_{AB}$. 

However, we soon encounter a problem, since none of these operators neither has $\vert n \rangle \langle m \vert$ as eigenstate, nor acts  as $\vert n \rangle \langle m \vert \mapsto n\,m\, \vert n-1 \rangle \langle m-1 \vert$, which is needed for the CS basis to recover the crutial property $a \vert \alpha \rangle = \alpha \vert \alpha \rangle$: 
\begin{subequations}
\begin{align}
\hat{T}_- \, \vert n \rangle \langle m \vert =&\, - \frac{1}{2}\, \vert n+1 \rangle \langle m+1 \vert - \frac{1}{2}\,  (n+m+1) \, \vert n \rangle \langle m \vert - \frac{1}{2}\, n\, m\, \vert n-1 \rangle \langle m-1 \vert, \\
\hat{T}_+ \, \vert n \rangle \langle m \vert =&\, \frac{1}{2}\, \vert n+1 \rangle \langle m+1 \vert, \\
\hat{T}_0 \, \vert n \rangle \langle m \vert = & \, \vert n+1 \rangle \langle m+1 \vert + \frac{1}{2}\, (n+m+1)\, \vert n \rangle \langle m \vert. 
\end{align}
\label{T_action}
\end{subequations}
In context of the sketched scheme, we construct linear combinations of these operators with the following properties: $\hat{\mathcal{T}}_-$ with the desired "annihilation" action $\hat{\mathcal{T}}_-\, \vert n \rangle \langle m \vert \propto n\,m\, \vert n-1 \rangle \langle m-1 \vert$; $\hat{\mathcal{T}}_0$ acting as the number operator (this will become useful in CS's), that is decomposable as $\hat{\mathcal{T}}_0 \propto \hat{N}(a) + \hat{N}(b)$, i.e. counting $n$'s and $m$'s indipendently; $\hat{\mathcal{T}}_+$ linearly independent of the two above. This leads us to the following 

\begin{prop}[$sl(2,\mathbb{C})$ isomorphism]
\label{iso}
Operators $\hat{\mathcal{T}}_- := \hat{T}_+ - \hat{T}_- - \hat{T}_0$, $\hat{\mathcal{T}}_+ := - \hat{T}_+$ and $\hat{\mathcal{T}}_0 := \hat{T}_0 - 2 \hat{T}_+$ comprise $sl(2,\mathbb{C})$ algebra isomorphic to (\ref{sl2_T}) and act on $\vert n \rangle \langle m \vert \in \mathcal{H}_A$ as follows:
\begin{equation}
\hat{\mathcal{T}}_- \, \vert n \rangle \langle m \vert = \frac{1}{2}\, n\,m\, \vert n-1 \rangle \langle m-1 \vert , \, \, \, \, \, \hat{\mathcal{T}}_+ \vert n \rangle \langle m \vert = -\frac{1}{2}\,  \vert n+1 \rangle \langle m+1 \vert , \, \, \, \, \, \hat{\mathcal{T}}_0 \, \vert n \rangle \langle m \vert = \frac{1}{2}\, (n+m+1)\, \vert n \rangle \langle m \vert. 
\end{equation}
\end{prop}
\begin{proof}
Direct computation using (\ref{sl2_T}) and (\ref{T_action}). 
\end{proof}

Coming to the second step, we would like to exponentiate $e^{\beta \mathcal{S}} \equiv e^{ \beta \, (\hat{\mathcal{T}}_- + \hat{\mathcal{T}}_+ )}$. As $[\hat{\mathcal{T}}-, \hat{\mathcal{T}}_+] \neq 0$ and both CS and $\vert n \rangle \langle m \vert$ are not eigenstates of $\hat{S}$, the idea of the Gauss decomposition is to be used: 

\begin{lemma}[Gauss decomposition for $\hat{\mathcal{T}}$'s]
Let $\hat{\mathcal{T}}_\pm, \hat{\mathcal{T}}_0$ be as above. Then for any $\beta\in \mathbb{R}$ the following holds:
\begin{equation}
e^{ \beta \, (\hat{\mathcal{T}}_- + \hat{\mathcal{T}}_+ )} = e^{2 \tanh \frac{\beta}{2} \,  \hat{\mathcal{T}}_+} \, e^{-2 \log \left( \cosh \frac{\beta}{2} \right) \, \hat{\mathcal{T}}_0} \, e^{2 \tanh \frac{\beta}{2} \, \hat{\mathcal{T}}_-}.
\end{equation}
\end{lemma}
\begin{proof}
Utilizing the $sl(2,\mathbb{C})$ isomorphism
\begin{equation}
2 \, \hat{\mathcal{T}}_+ \leftrightarrow X = \begin{pmatrix} 0 & 1 \\ 0 & 0 \end{pmatrix}, \, \, \, \, \, 2\,  \hat{\mathcal{T}}_0 \leftrightarrow H = \begin{pmatrix} 1 & 0 \\0 & -1 \end{pmatrix}, \, \, \, \, \,  2\,  \hat{\mathcal{T}}_- \leftrightarrow Y = \begin{pmatrix} 0 & 0 \\1 & 0 \end{pmatrix}
\end{equation}
we have 
\begin{equation}
e^{ \beta \, (\hat{\mathcal{T}}_- + \hat{\mathcal{T}}_+ )} \leftrightarrow \exp \left(\frac{\beta}{2} \, \begin{pmatrix}
0 & 1 \\ 1 & 0 \end{pmatrix}\right) = \begin{pmatrix} 
\cosh \frac{\beta}{2} & \sinh  \frac{\beta}{2} \\ \sinh  \frac{\beta}{2} & \cosh  \frac{\beta}{2} \end{pmatrix}
\end{equation}
on one side and 
\begin{equation}
e^{a \hat{\mathcal{T}}_+} \leftrightarrow \begin{pmatrix}
1 & a/2 \\ 0 & 1
\end{pmatrix} , \, \, \, e^{b \, \hat{\mathcal{T}}_0} \leftrightarrow \begin{pmatrix} 
e^{b/2} & 0 \\ 0 & e^{-b/2} \end{pmatrix}, \, \, \,  e^{c\, \hat{\mathcal{T}}_-} \leftrightarrow \begin{pmatrix}
1 & c/2 \\ 0 & 1
\end{pmatrix}
\end{equation}
on the other. Then matrix equation $e^{ \beta \, (\hat{\mathcal{T}}_- + \hat{\mathcal{T}}_+ )} = e^{a \hat{\mathcal{T}}_+} \, e^{b \, \hat{\mathcal{T}}_0} \, e^{c \, \hat{\mathcal{T}}_-}$ is condition on $a,b,c$ as functions of $\beta$ giving
\begin{equation}
a(\beta)= c(\beta) = 2\, \tanh \frac{\beta}{2},\, \, \, \, \, b(\beta) = - 2\, \log \left( \cosh \frac{\beta}{2} \right). 
\end{equation}
\end{proof}

Thanks to these two results we can calculate the matrix element $\langle \phi(\beta) \rangle_{CD}$. 

\begin{lemma}
\label{lemma_2}
Let $\hat{\mathcal{T}}_\pm, \hat{\mathcal{T}}_0$ be as above, let $\phi_{AB} = \vert A \rangle \langle B \vert$ be CS basis on $\mathcal{H}_A$ (with $\alpha$ suppressed), and let us denote $\phi(\beta) = e^{\beta\,  (\hat{\mathcal{T}}_- + \hat{\mathcal{T}}_+ )}\, \phi_{AB}$ and $\langle \phi(\beta) \rangle_{CD}= \langle C \vert \phi(\beta) \vert D \rangle$. Then 
\begin{equation}
\langle \phi(\beta) \rangle_{CD}=\frac{1}{\cosh \frac{\beta}{2}}\, e^{-\left( |A|^2 + |B|^2 + |C|^2 +|D|^2\right)/2} \, e^{\tanh \frac{\beta}{2}\, \left(A\,  \overline{B} - \overline{C} \, D \right)}\, e^{\left( A \, \overline{C} + \overline{B}\, D \right)/\cosh \frac{\beta}{2}}.\label{phiCD}
\end{equation}
\end{lemma}
\begin{proof}
The first exponential is trivial, as
\begin{align}
\hat{\mathcal{T}}_- \, \vert A \rangle \langle B \vert =&\, e^{-(|A|^2+ |B|^2)/2}\, \sum_{n,m=0}^\infty \frac{A^n \, \bar{B}^m}{n!\,m!}\, \hat{\mathcal{T}}_- \, \vert n \rangle \langle m \vert = e^{-(|A|^2+ |B|^2)/2}\, \sum_{n,m=0}^\infty \frac{A^n \, \bar{B}^m}{n!\,m!}\, \frac{1}{2}\, n\,m\, \vert n-1 \rangle \langle m-1 \vert  \notag \\
=&\,\frac{1}{2}\, A\, \bar{B}\, e^{-(|A|^2+ |B|^2)/2}\, \sum_{n,m=1}^\infty \frac{A^{n-1} \, \bar{B}^{m-1}}{(n-1)!\,(m-1)!}\,  \vert n -1\rangle \langle m -1\vert = \frac{1}{2}\, A\, \bar{B}\,\vert A \rangle \langle B \vert,
\end{align}
thus $e^{2 \tanh \frac{\beta}{2} \, \hat{\mathcal{T}}_-} \, \vert A \rangle \langle B \vert = e^{ \tanh \frac{\beta}{2}\, A \bar{B}} \, \vert A \rangle \langle B\vert$ yields only a prefactor. 

In the second one, rescaling of the CS occurs:
\begin{align}
e^{-2 \, \log \left( \cosh \frac{\beta}{2} \right) \, \hat{\mathcal{T}}_0} \, \vert A \rangle \langle B \vert =&\, e^{-(|A|^2+ |B|^2)/2}\, \sum_{n,m=0}^\infty \frac{A^n \, \bar{B}^m}{n!\,m!}\, \sum_{k=0}^\infty \, \frac{\left(-2 \,\log  \cosh \frac{\beta}{2} \right)^k}{k!} \,\hat{\mathcal{T}}_0^k\, \vert n \rangle \langle m \vert \notag \\
=&\, e^{-(|A|^2+ |B|^2)/2}\, \sum_{n,m=0}^\infty \frac{A^n \, \bar{B}^m}{n!\,m!}\, \sum_{k=0}^\infty \, \frac{\left(-2 \,\log  \cosh \frac{\beta}{2} \, \frac{n+m+1}{2}\right)^k}{k!} \,\vert n \rangle \langle m \vert \notag \\
=&\, e^{-(|A|^2+ |B|^2)/2}\, \sum_{n,m=0}^\infty \frac{A^n \, \bar{B}^m}{n!\,m!}\, e^{\left(- \log \cosh \frac{\beta}{2} \right) \,(n+m+1)} \vert n \rangle \langle m \vert \notag \\
=&\, \frac{1}{\cosh \frac{\beta}{2}}\, e^{-(|A|^2+ |B|^2)/2}\, \sum_{n,m=0}^\infty  \frac{\left(A/ \cosh \frac{\beta}{2}\right)^n \, \left(\bar{B}/ \cosh \frac{\beta}{2}\right)^m}{n!\,m!}\,\vert n \rangle \langle m \vert \notag \\
=&\, \frac{1}{\cosh \frac{\beta}{2}}\, e^{-\frac{1}{2}\, \tanh^2 \frac{\beta}{2}\, \left(|A|^2+|B|^2 \right)} \, \vert A' \rangle \langle B' \vert, 
\label{rescalingAB}
\end{align}
where $X' = X/ \cosh \frac{\beta}{2}$. Thus 
\begin{equation}
e^{-2 \log \left( \cosh \frac{\beta}{2} \right) \, \hat{\mathcal{T}}_0} \, e^{2 \tanh \frac{\beta}{2} \, \hat{\mathcal{T}}_-} \, \vert A \rangle \langle B \vert = \frac{e^{\tanh \frac{\beta}{2}\, A \bar{B}}\, e^{-\frac{1}{2}\, \tanh^2 \frac{\beta}{2}\, \left(|A|^2+|B|^2 \right)} }{\cosh \frac{\beta}{2}}\, \vert A' \rangle \langle B' \vert =: \mathcal{V}(\beta)_{AB} \,\vert A' \rangle \langle B' \vert 
\end{equation}
so far. 

As for the third one, no closed formula can be obtained:
\begin{align}
\phi(\beta) =&\, e^{2 \tanh \frac{\beta}{2} \,  \hat{\mathcal{T}}_+} \, \mathcal{V}_{AB}(\beta) \, \vert A' \rangle \langle B' \vert = \mathcal{V}_{AB}(\beta) \, e^{-(|A'|^2+ |B'|^2)/2}\, \sum_{n,m=0}^\infty  \frac{A^{\prime \,n} \, \bar{B}^{\prime \,m}}{n!\,m!}\, \sum_{k=0}^\infty \, \frac{\left(2\, \tanh \frac{\beta}{2} \right)^k}{k!} \,\hat{\mathcal{T}}_+^k\, \vert n \rangle \langle m \vert \notag \\
=& \, \, \mathcal{U}_{AB}(\beta)\, \sum_{n,m=0}^\infty  \frac{A^{\prime \,n} \, \bar{B}^{\prime \,m}}{n!\,m!}\, \sum_{k=0}^\infty \, \frac{\left(- \tanh \frac{\beta}{2} \right)^k}{k!} \, \vert n+k \rangle \langle m+k \vert. 
\end{align} 

Using the fact that $\langle n \vert m \rangle = n! \, \delta_{nm}$ we have non-zero dot product of the CS with the standard states: $\langle C \vert n+k \rangle = e^{-|C|^2/2}\, \sum_{l=0}^\infty \, \frac{C^l}{l!}\, \langle l \vert n+k \rangle = e^{-|C|^2/2}\, C^{n+k}$, similarly for the reversed order. Hence
\begin{align}
\langle \phi(\beta) \rangle_{CD} =&\, \mathcal{U}_{AB}(\beta)\,e^{-\left(|C|^2+|D|^2\right)/2}\, \sum_{n,m=0}^\infty  \frac{A^{\prime \,n} \, \bar{B}^{\prime \,m}}{n!\,m!}\, \sum_{k=0}^\infty \, \frac{\left(- \tanh \frac{\beta}{2} \right)^k}{k!} \, \bar{C}^{n+k}\, D^{m+k} \notag \\
 =&\, \mathcal{U}_{AB}(\beta) e^{-\left(|C|^2+|D|^2\right)/2}\, e^{A^\prime \, \bar{C} + \bar{B}^\prime \, D}\, e^{-\tanh \frac{\beta}{2}\, \bar{C}\,D} \notag \\
 \equiv &\, \frac{1}{\cosh \frac{\beta}{2}}\, e^{-\left( |A|^2 + |B|^2 + |C|^2 +|D|^2\right)/2} \, e^{\tanh \frac{\beta}{2}\, \left(A\,  \overline{B} - \overline{C} \, D \right)}\, e^{\left( A \, \overline{C} + \overline{B}\, D \right)/\cosh \frac{\beta}{2}}.
\end{align}
\end{proof}

Returning back to (\ref{HS_norm}) leads immediately to the following 
\begin{thm}[Boost invariance]
\label{thm_1}
The CS basis on $\mathcal{H}_A$ is invariant w.r.t. boost of $su(2,2)$ generated by $\hat{S}_{0i}$.
\end{thm}
\begin{proof}
Let us first carry on the calculation for $\langle \phi(\beta) \rangle_{CD}$ , i.e. $|| \phi(\beta) ||_{\mathcal{H}_A}^2 = \int dC\, dD\, |\langle \phi(\beta) \rangle_{CD}|^2$. Using Lemma \ref{lemma_2} we have
\begin{equation}
||\phi(\beta)||^2_{\mathcal{H}_A} = \int dC\, dD\, \frac{1}{\cosh^2 \frac{\beta}{2}}\, e^{-\left( |A|^2 + |B|^2 + |C|^2 +|D|^2\right)} \, e^{\tanh \frac{\beta}{2}\, \left(A\,  \overline{B} - \overline{C} \, D +h.c.\right)}\, e^{\left( A \, \overline{C} + \overline{B}\, D +h.c.\right)/\cosh \frac{\beta}{2}},
\end{equation}
where we imposed the $h.c.$ notion in an obvious way, $\bar{X} Y + h.c. \equiv \bar{X} Y + X \bar{Y} =  2\, (X_R Y_R + X_I Y_I )$ (subscripts $R,I$ for real and imaginary part respectively). Integrating over real and imaginary parts separately we get (recall the norm, cf. text  (\ref{coherent}) bellow):
\begin{align}
\dots =  \frac{1}{\cosh^2 \frac{\beta}{2}}\, e^{-\left( |A|^2 + |B|^2\right)} \, e^{\tanh \frac{\beta}{2}\, A\,  \overline{B}  +h.c.} \,&\, \int \frac{dC_R \, dD_R}{\pi} \, e^{-C_R^2 -D_R^2} \, e^{-2 \tanh \frac{\beta}{2}\,  C_R D_R} \, e^{\frac{2}{\cosh \frac{\beta}{2}} (C_R A_R + D_R B_R)} \cdot \notag \\
 \cdot &\, \int \frac{dC_I \, dD_I}{\pi} \, e^{-C_I^2 -D_I^2} \, e^{-2 \tanh \frac{\beta}{2}\, C_I D_I} \, e^{\frac{2}{\cosh \frac{\beta}{2}} (C_I A_I + D_I B_I)}.
\end{align}
Both real and imaginary integrals yield $\cosh \frac{\beta}{2}\, e^{A_\ast^2+B_\ast^2- 2 \tanh \frac{\beta}{2} A_\ast B_\ast}$ for $\ast = R,I$, which precisely compensates for the prefactor in the expression, thus $||\phi(\beta)||^2_{\mathcal{H}_A}=1$. 

Since the cases $\alpha=1$ and $\alpha=2$ are independent from the construction ($[\hat{\mathcal{T}}_\ast^{\bar{1}},\hat{\mathcal{T}}_\ast^{\bar{2}}]=0$), the same result applies for $\Psi(\beta) \equiv e^{\beta \, \hat{S}_{03}} \, | A_1, A_2 \rangle \langle B_1, B_2 \vert$. Eventually, we chose the $3$-direction in the boost, but the result does not depend on any direction, thus rotation to general direction is trivial as well.  
\end{proof}

\begin{remark}
On account of this theorem, one can compute the norm of the boosted state in the former (inconvenient) basis $\vert n_1, n_2 \rangle$ exactly. This is done in Appendix \ref{sec_ap_A}, where we arrive at the same result. In light of the presented calculation, we will be able to derive a non-trivial identity among hypergeometric functions, cf. Proposition \ref{parseval}. 
\end{remark}

\subsection{Classical fuzzy fields with $m=0$ and $\alpha=j_1-j_2$}
The oscillator construction above enables us to construct first class of classical fuzzy fields being an irreducible representation of $su(2,2)$ (fuzzy) algebra. To do so, we 
\begin{enumerate}[label=(\roman*)]
\item construct the Casimir oprators for the fundamental and dual representation restricted to subalgebra $so(1,3)$, i.e. Poincar\'e invariants;
\item construct the Casimir operators for the both representations of $so(2,4)$;
\item find the spectrum of $\hat{S}_{05}$. 
\end{enumerate}
Then following \cite{Ma} we will be able to identify the class of the representations. 

Starting with $so(1,3)$ subalgebra, we calculate the spectrum of $\hat{P}^2$ and $\hat{W}^2$, where $W_\mu= \frac{1}{2}\, \epsilon_{\mu \nu \rho \sigma} \, S^{\nu \rho} \, P^\sigma$ is the Pauli-Lubanski vector. The first invariant is
\begin{align}
\hat{P}_\mu \hat{P}^\mu =&\, (\hat{P}_0)^2- \delta^{ij} \hat{P}_i \hat{P}_j \notag \\
=&\, - \frac{1}{4} \left[ (\hat{a}^\dagger - \hat{b}^\dagger )( \hat{a} - \hat{b} )(\hat{a}^\dagger - \hat{b}^\dagger)( \hat{a} - \hat{b}) - \delta^{ij} (\hat{a}^\dagger - \hat{b}^\dagger )\sigma_i ( \hat{a} - \hat{b} )(\hat{a}^\dagger - \hat{b}^\dagger)\sigma_j( \hat{a} - \hat{b}) \right].
\end{align} 
We will work out explicitly the first term proportional to four $\hat{a}$'s. Writing it in indices gives 
\begin{align}
\hat{a}^\dagger \hat{a}\,  \hat{a}^\dagger \hat{a} - \delta^{ij} \hat{a}^\dagger \sigma_i \hat{a} \, \hat{a}^\dagger \sigma_j \hat{a} = \hat{a}^{\dagger \alpha} \hat{a}_\alpha \, \hat{a}^{\dagger \beta} \hat{a}_\beta - \delta^{ij} \, \hat{a}^{\dagger \alpha} \varsigma_{\alpha \dot{\alpha}} \sigma_i^{\dot{\alpha} \beta} \hat{a}_{\beta} \,\hat{a}^{\dagger \gamma} \varsigma_{\gamma \dot{\gamma}} \sigma_i^{\dot{\gamma} \beta} \hat{a}_{\delta} .
\end{align}
We use the completness relation among the Pauli matrices ($
\delta^{ij} \sigma_i^{\dot{\alpha} \beta} \sigma_j^{\dot{\gamma} \delta} = 2 \varsigma^{\dot{\alpha} \delta} \varsigma^{ \dot{\gamma} \beta} - \varsigma^{\dot{\alpha} \beta} \varsigma^{\dot{\gamma} \delta}$) and the commutator relations on $\mathcal{H}_A$ to get 
\begin{align}
\dots = \hat{a}^{\dagger \alpha} \hat{a}_\alpha \, \hat{a}^{\dagger \beta} \hat{a}_\beta - 2 \varsigma^{\dot{\alpha} \delta} \varsigma^{ \dot{\gamma} \beta} \varsigma_{\alpha \dot{\alpha}}  \varsigma_{\gamma \dot{\gamma}} \, \hat{a}^{\dagger \alpha} \hat{a}_\beta \, \hat{a}^{\dagger \gamma} \hat{a}_\delta +
\varsigma^{\dot{\alpha} \beta} \varsigma^{\dot{\gamma} \delta} \varsigma_{\alpha \dot{\alpha}}  \varsigma_{\gamma \dot{\gamma}}\, \hat{a}^{\dagger \alpha} \hat{a}_\beta \, \hat{a}^{\dagger \gamma} \hat{a}_\delta = - 2 \hat{a}^\alpha \, \hat{a}_\alpha.
\end{align}
The rest of the calculation is analogous and the complete result reads $\hat{P}^2=0$, i.e. we deal with a massless representation. One can check in a similar way, that also $\hat{W}^2=0$, which points towards the fact that $W_\mu$ is parallel to $P_\mu$, hence representations with helicity. Then the proper invariant is $W_0^2 = W_3^2$. In this case, $W_0 = - J_i\, P_i$ with $J_i = \frac{1}{2}\, \epsilon_{ijk} \, S_{jk}$. Direct action of $\hat{W}_0=- \hat{J}_i \, \hat{P}_i$ on the basis in $\mathcal{H}_A$ yields
\begin{equation}
\hat{W}_0\, \vert n_1, n_2 \rangle \langle m_1, m_2 \vert = n_1 \, m_1\, \frac{\kappa-2}{4}\, \vert n_1-1, n_2 \rangle \langle m_1-1, m_2 \vert + n_2\, m_2 \, \frac{\kappa+2}{4}\, \vert n_1, n_2-1 \rangle \langle m_1, m_2-1 \vert, 
\end{equation}
or in CS basis 
\begin{equation}
\hat{W}_0\, \vert A_1,A_2 \rangle \langle B_1, B_2 \vert = \left(\frac{\kappa-2}{4}\, A_1 \, \bar{B}_1 + \frac{\kappa+2}{4}\, A_2\, \bar{B}_2 \right) \vert A_1,A_2 \rangle \langle B_1, B_2 \vert.
\end{equation}
Thus there is a direct connection between the helicity (or spin in potential massive case) and the chiral parameter. 

Moving to the full conformal algebra, we uncovert even greater importance of the chiral parameter.  Recall that the three Casimir operators are
\begin{equation}
C_2 =  \frac{1}{2}\, S^{ab}\, S_{ab}, \, \, \, C_3 =   \frac{1}{3!}\, \varepsilon_{abcdef} \, S^{ab} \,S^{cd}\, S^{ef}, \, \, \, 
C_4 =  \frac{1}{2}\, S^{ab} \,S_{bc} \,S^{cd}\, S_{da} .
\label{casimirs} 
\end{equation}
They can be recasted in a more convenient form (\cite{LaLa}) using the conformal analogues of Pauli-Lubanski vector ($V_\mu \propto \epsilon_{\mu \nu \rho \sigma} \, S^{\nu \rho} \, S^{\sigma 4}$ and $U_\mu \propto \epsilon_{\mu \nu \rho \sigma} \, S^{\nu \rho} \, S^{\sigma 5}$), but it is not of a great benefit for us. Instead, we use NCAlgebra package (\cite{NC}) to conduct the computation. We let it express the Casimir operators in terms of $\hat{a}$'s and $\hat{b}$'s first and then subtract the powers of $\hat{C}_1$, as Casimir operators in fundamental representation need to be proportional to the central element (and its powers). The resulting three operators (plus the central element for reference) with their eigenvalues on $\vert n_1, n_2 \rangle \langle m_1, m_2 \vert$ basis are
\begin{subequations}
\begin{align}
\hat{C}_1 =&\, \frac{1}{2} \left( \hat{a}^\dagger \hat{a} - \hat{b}^\dagger \hat{b} \right), \, \, \, \, \, \, \, \, \, \, \, \, \, \, \, \, \, \, \, \, \, \, \, \, \, \, \, \, \, \, \, \, \, \, \, \, \, \, \, \, \, \lambda_1= \frac{1}{2}\, (\kappa-2), \\
\hat{C}_2=&\, -3\, \hat{C}_1^2- 6\, \hat{C}_1, \, \, \, \, \, \, \, \, \, \, \, \, \, \, \,\, \, \, \, \, \, \, \, \, \, \, \, \, \, \, \, \, \, \, \, \, \, \, \, \,\,\, \lambda_2 = -\frac{3}{4}\,(\kappa-2)\, (\kappa+2), \\
\hat{C}_3=&\, 8\,i\,\hat{C}_0^3+ 24\,i\,\hat{C}_0^2+16\,i\,\hat{C}_0, \, \, \, \, \, \, \, \, \, \, \, \, \, \, \,  \lambda_3=i\, \kappa\,(\kappa-2)\, (\kappa+2), \\
\hat{C}_4=&\, 3\,\hat{C}_1^4+12\,\hat{C}_1^3+24\,\hat{C}_1^2+24\,\hat{C}_1, \, \, \, \lambda_4 =\frac{3}{16}\, (\kappa^2+12)\, (\kappa-2)\, (\kappa+2).
\end{align}
\end{subequations}

\begin{remark}
As for the dual representation, the same procedure with dual generators  yields $\tilde{C}_2 = \hat{C}_2$, $\tilde{C}_3 = - \hat{C}_3$ and $\tilde{C}_4 = \hat{C}_4$ with hats changed for tildes. 
\end{remark}

Final step before the representation class identification is to find the spectrum of $\hat{S}_{05}$. As $\hat{S}_{05} = -i \hat{T}_0$, we use the isomorphism (cf. Proposition \ref{iso}) and immediately have
\begin{equation}
-i\, \hat{\mathcal{T}}_0 \, \vert n_1, n_2 \rangle \langle m_1, m_2 \vert = -\frac{i}{2} \left(n_1+n_2+m_1+m_2+2 \right)\, \vert n_1, n_2 \rangle \langle m_1, m_2 \vert,
\end{equation} 
hence the spectrum is $\sigma(i \, \hat{S}_{05}) = \{1+\frac{n_1+n_2+m_1+m_2}{2}\, \vert \, n_1, n_2, m_1, m_2 \in \mathbb{N}_0 \}$.

We see that starting from fundamental (or dual) representation of fuzzy $su(2,2)$ enables us to construct an irreducible representation describing massless field with helicity connected to the chiral parameter. Moreover, the conformal properties of the field are given uniquely in terms of the same parameter, where the expression $(\kappa-2)/2= - \left(1 -\kappa/2 \right)$ plays a special role.  These observations lead us to prove the following 
\begin{thm}[Classical fuzzy massless field]
\label{massless}
For each $\kappa\in \mathbb{N}_0$ there is a classical fuzzy massless field with helicity $j = -\kappa/2$, i.e. a $su(2,2)\cong so(2,4)$ invariant subspace $\mathcal{A}_{(j,0)}^0 \subset \mathcal{A}$ (or equivalently $\mathcal{A}_{(0,j)}^0$) given as $\mathcal{A}_{(j,0)}^0= \{ (a^{\dagger 1})^{n_1}\, (a^{\dagger 2})^{n_2}\, (a_1)^{m_1}\, (a_2)^{m_2} \,  \vert \, n_1, n_2 \in \mathbb{N}_0,  (n_1+n_2)-(m_1+m_2) = \kappa = 2j\}$ . 
\end{thm}
\begin{proof}
Since $\hat{P}^2=0$ for the fundamental representation, in analogy with Mack's representation classification (\cite{Ma}) we immediately see, that (5) is our case, i.e. representation labeled with $j_1 j_2=0$, $d=j_1+j_2+1$ that contains fields with $m=0$ and helicity $j_1-j_2$. In order to specify $d$ and $j_{1,2}$, we follow Lemma 2 ibid. We are thus interested in operator $T(\gamma) = e^{-2 \, \pi\, n \, \hat{S}_{05}}$ (note the mathematical convention), whose spectrum is $\omega(\gamma) = e^{2\,\pi\,i\,n\,\left(1+\frac{n_1+n_2+m_1+m_2}{2}\right)} = e^{2\,\pi\,i\,n\,\left( 1 + (m_1+m_2) + \frac{\kappa}{2}\right)} $. Then the lowest  spectral value is $d= 1+\kappa/2$. Without loss of generality take $j_2=0$, $j_1\neq 0$. Then $\alpha= j_1=\kappa/2$, i.e. the chiral parameter is double the helicity of the field, which manifests the crutial importance of the $\kappa$. Finally, since $2j_1 \in \mathbb{N}_0$, we have $\kappa \in \mathbb{N}_0$. 
\end{proof}
\begin{remark}
The same result applies also for the dual representation. We thus have two unitary inequivalent irreducible representations of $su(2,2)$, so called short (doubleton) representations.
\end{remark}

\bigskip

\section{Product of Two Doubleton Representations}
It is quite surprising that one can construct one more class of irreducible representations out of the formerly analysed (two unitary inequivalent copies of) doubleton representation. 

\subsection{The Idea and the Ansatz}
Let us construct direct product of fundamental and dual representation in the standard way: 

\begin{definition} Let $\hat{S}_{ab}$ and $\tilde{S}_{ab}$ be generators of the $su(2,2) \cong so(2,4)$ algebra in fundamental (on $\mathcal{H}_A$) and dual (on $\mathcal{H}_A^\prime$) representation respectively. Then the direct product of these representations is the representation $\mathbf{S}_{ab} = \hat{S}_{ab} + \tilde{S}_{ab} \equiv \hat{S}_{ab} \otimes \id_{\mathcal{H}_A^\prime} + \id_{\mathcal{H}_A} \otimes \tilde{S}_{ab}$ acting on $\mathcal{H}_A \otimes \mathcal{H}_A^\prime = \spa \{ \vert n_1, n_2 \rangle \langle m_1, m_2 \vert \otimes \vert n_1^\prime , n_2^\prime \rangle \langle m_1^\prime, m_2^\prime \vert \}$. 
\end{definition}
\begin{remark}[Notation]
We will henceforth stick to the notation, where
\begin{enumerate}[label=(\roman*)]
\item the fundamental representation is constructed from  $S_{ab}$, C/A operators ($\hat{a}$'s and $\hat{b}$'s) and generators ($\hat{S}_{ab}$) are labeled by hat, and fuzzy functions $\Psi$ on $\mathcal{A}$ have no mark;
\item the dual representation is constructed from  $S'_{ab}$, C/A operators ($\tilde{a}$'s and $\tilde{b}$'s) and generators ($\tilde{S}_{ab}$) are labeled by tilde, and fuzzy functions $\Psi^\prime$ on $\mathcal{A}^\prime$ are primed;
\item the direct product of these representations has generators ($\mathbf{S}_{ab}$) and fuzzy functions $\boldsymbol{\Psi}$ on $\mathcal{A} \otimes \mathcal{A}^\prime$ in bold.
\end{enumerate}
\label{notation_hattilde}
\end{remark}
\begin{remark}
Recall (\cite{Fe}, $\S$12.4) that for two representations $(\rho_1, V_1)$ and $(\rho_2, V_2)$ of the same group $G$ we define the direct product of the representations (for $g \in G$) as $(\rho_1 \otimes \rho_2)(g):=\rho_1(g)\otimes \rho_2(g)$ on $V_1 \otimes V_2$ and the direct sum of representations as $(\rho_1 \oplus \rho_2)(g) := \rho_1(g) \oplus \rho_2(g)$ on $V_1 \oplus V_2$.   The corresponding derived representations of the algebra $\mathcal{G}$ are then (for $X \in \mathcal{G}$) as follows: $(\rho_1 \otimes \rho_2)^\prime(X) = \rho_1^\prime(X) \otimes \id_{V_2} + \id_{V_1} \otimes \rho_2^\prime(X)$ on $V_1 \otimes V_2$ for the product and $(\rho_1 \oplus \rho_2)^\prime(X)=\rho_1^\prime(X) \oplus \rho_2^\prime(X)$ on $V_1 \oplus V_2$ for the sum.  
\end{remark}

We now construct a special Ansatz for $\boldsymbol{\Psi}$, that will require special class of functions, calculated below. 

\begin{prop}[Massive Ansatz]
Let $p_\mu^\pm=(\epsilon, 0,0,\pm \epsilon)$ for some $\epsilon \geq 0$. Then if there exist functions $\Psi^\pm(a,a^\dagger) \in \mathcal{A}$, $\Psi^{\prime \, \pm}(a^\prime, a^{\prime\, \dagger}) \in \mathcal{A}^\prime$ such that $\hat{P}_\mu \Psi^\pm = p_\mu^\pm \Psi^\pm$ and $\tilde{P}_\mu \Psi^{\prime\,  \pm} = p_\mu^\pm \Psi^{\prime \, \pm}$, i.e. two massless fuzzy functions corresponding to fundamental and dual representation respectively, then $\boldsymbol{\Psi}^\pm(a, a^\dagger, a^\prime, a^{\prime\, \dagger}) = \Psi^\pm(a,a^\dagger)\, \Psi^{\prime\,\mp}(a^\prime, a^{\prime\, \dagger})$ satisfies $\mathbf{P}_\mu \, \boldsymbol{\Psi}^\pm = \mathbf{p}_\mu \, \boldsymbol{\Psi}^\pm$ for $\mathbf{p}_\mu =(2\epsilon,0,0,0)$, i.e. it is a fuzzy massive function. 
\end{prop}
\begin{proof}
Evident from the construction and $(\epsilon, 0,0,\pm \epsilon) + (\epsilon, 0,0,\mp \epsilon) = (2\epsilon,0,0,0)$. As $m=2\epsilon$ is the mass, we have $m\geq 0$. 
\end{proof}

The construction above works provided we can find the eigenfunctions of $\hat{P}_\mu$ (or analogously $\tilde{P}_\mu$). Success in finding such is summarized in the following 

\begin{lemma}[Eigenfunctions of $\hat{P}_\mu$]
Let $p_\mu^\pm=(\epsilon, 0,0,\pm \epsilon)$ for some $\epsilon \geq 0$. Then for every $\kappa \in \mathbb{Z}$ there exists a pair of functions $\Psi^\pm(a,a^\dagger) \in \mathcal{A}$ such that $\hat{P}_\mu \Psi^\pm = p_\mu^\pm \Psi^\pm$. 
\end{lemma}
\begin{proof}
Observe, that $(\hat{a}-\hat{b}) = [a,\,  \cdot \,\,]$ and $(\hat{a}^\dagger - \hat{b}^\dagger) = [a^\dagger,\, \cdot\,\,]$. Then for any $\Psi \in \mathcal{A}$ (cf. \ref{gen_osc})
\begin{subequations}
\begin{align}
\hat{P}_0 \, \Psi =&\,- \frac{i}{2}\, \left( \big[a^{\dagger 1}, [a_1, \Psi]\big] + \big[a^{\dagger 2}, [a_2, \Psi]\big] \right), \\
\hat{P}_3 \, \Psi =&\, - \frac{i}{2}\, \left( \big[a^{\dagger 1}, [a_1, \Psi]\big] - \big[a^{\dagger 2}, [a_2, \Psi]\big] \right)
\end{align} 
\end{subequations}
and $\hat{P}_{1,2} \, \Psi = 0$. This motivates separating $\alpha=1$ from $\alpha=2$, $\Psi(a,a^\dagger)= \psi(a_1, a^{\dagger 1})\,\psi(a_2,a^{\dagger 2})$. Then the following choice
\begin{subequations}
\begin{align}
\Psi^+:&\, \, \, \psi_2=1, \, \, \, -\frac{i}{2}\,\big[ a^{\dagger 1}, [a_1, \psi_1]\big] = \epsilon\, \psi_1, \\
\Psi^-:&\, \, \, \psi_1=1,\, \, \, -\frac{i}{2}\,\big[ a^{\dagger 2}, [a_2, \psi_2]\big] = \epsilon\, \psi_2
\end{align}
\end{subequations}
will be the solution to the eigenvalue problem. 

To find a solution $\psi_\alpha(a_\alpha,a^{\dagger \alpha})$ of
\begin{equation}
[a^{\dagger \alpha}, [a_\alpha, \psi_\alpha]] = 2 \, i \, \epsilon\, \psi_\alpha, \label{eigen_sep}
\end{equation} 
let us design for any $\kappa \in \mathbb{Z}$
\begin{equation}
\psi(a_\alpha,a^{\dagger \alpha}) = (a^{\dagger \alpha})^\kappa \, : \chi(N_\alpha) :,
\end{equation} 
where $N_{\alpha} = \hat{a}^{\dagger \alpha} \hat{a}_\alpha$ and $: \chi : $ is the normal ordering, cf. ftnt \ref{normal}. In this whole construction, no summation over $\alpha$ occurs. 

Inserting this Ansatz into (\ref{eigen_sep}),  realising that $\colon \hat{N}_{\bar{\alpha}}^k \colon \equiv \colon \underbrace{ \hat{a}^{\dagger \alpha} \hat{a}_\alpha \,\dots \, \hat{a}^{\dagger \alpha} \hat{a}_\alpha}_{k- \mbox{times}}\, \colon = (\hat{a}^{\dagger \alpha})^k \, (\hat{a}_{\alpha})^k $ and using (\ref{a_der}) yields 

\begin{align}
\big[a^{\dagger \alpha}, [a_\alpha, (a^{\dagger \alpha})^\kappa \, : \chi(N_\alpha): ] \big] =&\,  \big[a^{\dagger \alpha}, \kappa \, (a^{\dagger \alpha})^{\kappa-1} \, : \chi(N_\alpha) : + (a^{\dagger \alpha})^\kappa  : \frac{\partial \chi(N_\alpha)}{\partial N_\alpha} : \, a \big] \notag \\
=&\, (\kappa-1) \, (a^{\dagger \alpha})^\kappa\, : \chi^\prime(N_\alpha) : + (a^{\dagger \alpha})^\kappa\, a^{\dagger \alpha}\, : \chi^{\prime \prime}(N_\alpha) : \, a  \notag \\
=&\, (a^{\dagger \alpha})^\kappa\, : \big[ (\kappa-1) \, \chi^\prime(N_\alpha) + N_\alpha\, \chi^{\prime \prime}(N_\alpha) \big] :\,  \overset{\underset{\mathrm{!}}{}}{=} 2\,i\,\epsilon\, (a^{\dagger \alpha})^\kappa\, :\chi(N_\alpha) : \, .
\end{align}
Then (\ref{eigen_sep}) is recasted into 
\begin{equation}
: \,  N_\alpha\, \chi^{\prime \prime}(N_\alpha)+(\kappa-1) \, \chi^\prime(N_\alpha)  - 2\,i\,\epsilon \chi(N_\alpha) : \, = 0. 
\end{equation}
Thus we need to find a solution to the Bessel equation $x \, y''(x) + (\kappa-1) \, y'(x) - 2\,i\,\epsilon \, y(x) =0$, understand $y(x)$ as a formal power series in $x$ and take $x=N_\alpha$, and finally normal-order it. 

General solution for $\kappa \in \mathbb{Z}$ is 
\begin{equation}
y(x) = c_1 \, (8\,i\, \epsilon \,x)^{1-\kappa/2} \, I_{\kappa-2}\left((2+2i) \, \sqrt{\epsilon\,x} \right) + c_2\, (8\,i\, \epsilon \,x)^{1-\kappa/2} \, K_{\kappa-2}\left((2+2i) \, \sqrt{\epsilon\,x} \right).
\end{equation}
We take the first term (regular around $x=0$ for the sake of formal power series) and use $(\sqrt{8\,i})_0 = 2+2i$ to get 
\begin{equation}
\chi(N_\alpha) = (8\,i\,\epsilon\, N_\alpha)^{1-\kappa/2}\, I_{\kappa-2}(\sqrt{8\,i\,\epsilon\, N_\alpha}). 
\end{equation}
Hence the solution to the separated eigenvalue problem is 
\begin{equation}
\psi_\alpha(a_\alpha,a^{\dagger \alpha}) = (a^{\dagger \alpha})^\kappa \, : \, (8\,i\,\epsilon\, N_\alpha)^{1-\kappa/2}\, I_{\kappa-2}(\sqrt{8\,i\,\epsilon\, N_\alpha}) \, : \, 
\end{equation}
and the full solution is obtained, when one takes $\alpha=1$ and $\alpha=2$ for $\Psi^+$ and $\Psi^-$ respectively. 
\end{proof}

\begin{remark}
For the dual case, the situation is completely analogous after exchanging hats for tildes and no-primed for primed quantities. 
\end{remark}
\begin{remark}
Note, that we chose $p_\mu^\pm$ conventionally pointing into $3$-direction, but one can rotate the result into any direction. 
\end{remark}

We can use this result to derive a more useful formula for $\psi_\alpha$, namely the result for the action of  $\psi_\alpha (\in \mathcal{A})$ on $\vert n_\alpha \rangle \in \mathcal{H}_F$. 
\begin{cor}
$\psi_\alpha(a_\alpha,a^{\dagger \alpha}) \, \vert n_\alpha \rangle =  \sqrt{\frac{(n_\alpha + \kappa)!}{n_\alpha!}}\, \frac{2^{2-\kappa}}{\Gamma(\kappa-1)}\, _1F_1(-n_\alpha; \kappa-1; -2\,i\,\epsilon) \, \vert n_\alpha \rangle$ on $\mathcal{H}_F$.
\end{cor}
\begin{proof}
Taking power series in $N_\alpha$ yields
\begin{align}
\psi_\alpha(a_\alpha,a^{\dagger \alpha}) =&\,  (a^{\dagger \alpha})^\kappa\, : (8\,i\,\epsilon)^{1-\kappa/2} \, \sum_{m=0}^\infty \frac{1}{m!\, \Gamma(m+\kappa-1)}\,\left( \frac{\sqrt{8\,i\,\epsilon\, N_\alpha}}{2}\right)^{2m+\kappa-2} \, : \notag \\
=&\, (a^{\dagger \alpha})^\kappa\, \sum_{m=0}^\infty \, \frac{2^{2-\kappa}}{m!\, \Gamma(m+\kappa-1)}\, : \, (2\,i\,\epsilon\, N_\alpha)^m \, : \, . 
\end{align}
Direct calculation (cf. \ref{norm_FF}) reveals that $\colon \hat{N}_\alpha^k \colon \, |n_\alpha \rangle =
\begin{cases} 
\frac{n_\alpha!}{(n_\alpha -k)!} \,|n_\alpha \rangle, & k \le n_\alpha, \\
0, & k>n_\alpha.
\end{cases}$, thus
\begin{align}
\psi_\alpha(a_\alpha,a^{\dagger \alpha}) \, \vert n_\alpha \rangle =&\, (a^{\dagger \alpha})^\kappa\, \sum_{m=0}^{n_\alpha} \, \frac{2^{2-\kappa}}{m!\, \Gamma(m+\kappa-1)}\, (2\,i\,\epsilon)^m \, \frac{n_\alpha!}{(n_\alpha-m)!}\, \vert n_\alpha \rangle \notag \\
=&\, (a^{\dagger \alpha})^\kappa\, \frac{2^{2-\kappa}}{\Gamma(\kappa-1)}\, _1F_1(-n_\alpha; \kappa-1; -2\,i\,\epsilon) \, \vert n_\alpha \rangle \notag \\
=&\, \sqrt{\frac{(n_\alpha + \kappa)!}{n_\alpha!}}\, \frac{2^{2-\kappa}}{\Gamma(\kappa-1)}\, _1F_1(-n_\alpha; \kappa-1; -2\,i\,\epsilon) \, \vert n_\alpha +\kappa \rangle.
\end{align}
\end{proof}

Using the above results, we close this subsection by summarizing, that we found two functions $\boldsymbol{\Psi}^+ = \psi_1(a_1, a^{\dagger 1}) \, \psi_2^\prime(a_2^\prime, a^{\prime \, \dagger 2})$ and $\boldsymbol{\Psi}^- = \psi_2(a_2, a^{\dagger 2}) \, \psi_1^\prime(a_1^\prime, a^{\prime \, \dagger 1})$ such that $\mathbf{P}^2 \, \boldsymbol{\Psi}^\pm = m^2 \boldsymbol{\Psi}^\pm$ provided $\psi(a, a^\dagger) = (a^\dagger)^\kappa \, (4\,i\,m \, N)^{1-\kappa/2} \, I_{\kappa-2} \left( \sqrt{4\,i\,m\,N}
\right)$ for any $\kappa$. Such fuzzy functions  are represented on $\mathcal{A} \otimes \mathcal{A}^\prime$ with the action on $\vert n_1, n_2 \rangle \otimes \vert n_1^\prime, n_2^\prime \rangle \in \mathcal{H}_F \otimes \mathcal{H}_F^\prime$ as follows:
\begin{align}
\boldsymbol{\Psi}^+ \, \vert n_1, n_2 \rangle \otimes \vert n_1^\prime, n_2^\prime \rangle =&\, \sqrt{\frac{(n_1+\kappa)! \, (n_2^\prime + \kappa^\prime)!}{n_1! \, n_2^\prime!}}\, \frac{2^{4-\kappa-\kappa^\prime}\, _1F_1(-n_1; \kappa-1; -i \, m) \, _1F_1(-n_2^\prime; \kappa^\prime-1;-i\,m)}{(\kappa-2)!\, (\kappa^\prime-2)!} \notag \\
&\, \cdot \vert n_1+\kappa, n_2 \rangle \otimes \vert n_1^\prime, n_2^\prime+\kappa^\prime \rangle ,
\end{align}  
$\boldsymbol{\Psi}^-$ analogously. 

\subsection{Classical fuzzy fields with $m>0$ and $s=j_1+j_2$}
The construction from the previous subsection is only one step from identifying another class of irreducible representations according to Mack. Thus in complete analogy with the doubleton representation, we need to find spectrum of $\mathbf{S}_{05}$. This is trivial and we have 
\begin{equation}
\sigma(i\, \mathbf{S}_{05}) = \Big\{ \left(1+\frac{n_1+n_2+m_1+m_2}{2}\right) + \left(1+\frac{n_1^\prime + n_2^\prime +m_1^\prime + m_2^\prime}{2} \right) \, \Big| \, n_\alpha, m_\alpha, n_\alpha^\prime, m_\alpha^\prime \in \mathbb{N}_0 \Big\}. 
\end{equation}
Using $\kappa = (n_1+n_2)-(m_1+m_2)$ and $\kappa^\prime = (n_1^\prime+n_2^\prime)-(m_1^\prime+m_2^\prime)$ we immediately see that $d= 2+ \frac{\kappa+\kappa^\prime}{2}$. Moreover, the fuzzy functions have $m>0$, thus we are in the class (4) of Mack's classification. Then $d=2+j_1+j_2$ yields (without loss of generality) $j_1=\kappa/2$, $j_2=\kappa^\prime/2$  and the total spin $s=j_1+j_2 =  \frac{\kappa+\kappa^\prime}{2}$. We thus proved the following 

\begin{thm}[Classical fuzzy massive field]
\label{massive}
For each $\kappa, \kappa^\prime =\mathbb{N}_0$ there is $\mathcal{A}_{(j_1,j_2)}^{m(\pm)} \subset \mathcal{A} \otimes \mathcal{A}^\prime$  given as $\mathcal{A}_{(j_1,j_2)}^{m(+)}= \big \{ (a^{\dagger\,1})^\kappa  : (4\,i\,m\,N_1)^{1-\kappa/2} I_{\kappa-2}(\sqrt{4\,i\,m\,N_1}) : (a^{\prime \, \dagger \, 2})^{\kappa^\prime}: (4\,i\,m\,N^\prime_2)^{1-\kappa^\prime/2} I_{\kappa^\prime-2}(\sqrt{4\,i\,m\,N^\prime_2}) :   \big\}$, i.e. constructed from $\boldsymbol{\Psi^+}$,  or equivalently $\mathcal{A}_{(j_1,j_2)}^{m(-)}$ from $\boldsymbol{\Psi^-}$. Then a classical fuzzy massive field with mass $m>0$ and spin $s=  \frac{\kappa+\kappa^\prime}{2}$ is a set of all such $\mathcal{A}_{(j_1,j_2)}^{m(\pm)}$ with $s=j_1+j_2$ forming a $su(2,2)\cong so(2,4)$ invariant subspace.   
\end{thm}

\bigskip

\section{Discussion and Outlook}
\label{sec_discussion}
We constructed massless and massive fuzzy fields with half-integer dimension $d$ as unitary irreducible representations of the non-commutative conformal algebra in four dimensions. We uncovered that the most central object is the chiral parameter $\kappa$ occuring at several stages of our construction. In light of \cite{KoPP}, the chiral parameter can be viewed as the charge of the magnetic monopole. This offers a new interpretation of our results: magnetic monopoles are responsible for the NC structure of the relativistic fuzzy space. 

Apart from this observation, our construction has several benefits:
\begin{enumerate}[label=(\roman*)]
\item It is explicit in contrary to \cite{Ma}, e.g. for a massless fuzzy field one takes $\mathcal{A}_{(j,0)}^0$ (cf. Theorem \ref{massless}) and for the massive fuzzy field with spin $s$ one takes all $\mathcal{A}_{(j_1,j_2)}^{m(\pm)}$ such that $j_1+j_2=s$ (cf. Theorem \ref{massive}). Moreover, all calculations are straightforward due to well-known manipulations on (modified) Hilbert spaces. 
\item It is more formal than the previous papers on similar topics ( \cite{GaKoPP}, \cite{KoPP}) when regarding parts with $\mathcal{H}_F$.
\item It presents novel feauture (apart from the main results of the paper): the auxiliary Hilbert space $\mathcal{H}_A$ as an alternative way of describing the action of $\hat{a}$'s and $\hat{b}$'s on $\Psi \in \mathcal{A}$. 
\end{enumerate}

However, our construction lacks the ability to be applicable also for the remaining two classes of unitary irreducible representations, i.e. classes (2) and (3) in \cite{Ma}. This is due to the fact, that the spectrum of the conformal Hamiltonian ($\hat{S}_{05}$) is discrete in our case (or more precisely half-integer). A possible way out of this problem could be the use of $q$-deformed C/A's, that would be presumably reflected also in the spectrum of $\hat{S}_{05}$. It is noteworthy that this idea could be the analogue of the $\mathcal{F}_\gamma$-module in \cite{GuVo}. 

Apart from the $q$-deformed oscillators, more prospects on possible extensions of our work could be made, both inspired by \cite{GuVo}: We used only fundamental (and dual) representation for the C/A's. It could be meaningful to introduce colours as well; moreover we used only bosonic C/A's, thus the introduction of fermionic oscillators could lead to the construction of superalgebra $su(2,2\vert \mathcal{N})$. All ideas, however, need to be analysed properly. 

Yet another, more physical, outlook can be made. If we identify $E_{(\mu)}^0 = \hat{P}_\mu + \tilde{P}_\mu$ and $E_{(\mu)}^3= \hat{P}_\mu - \tilde{P}_\mu$, that are quadratic separately in $\hat{a}, \hat{a}^\dagger, \hat{b}, \hat{b}^\dagger$ and $\tilde{a},\tilde{a}^\dagger, \tilde{b}, \tilde{b}^\dagger$ respectively, we see that $\left( E_{(\mu)}^0\right)^2 = - \left(E_{(\mu)}^3\right)^2 \ge 0$. If we thus added $E_{(\mu)}^{1,2}$ such that they sqare to zero and are quadratic in mixed combinations of hatted and tilded operators, we would immediately have the connection to gravity via the Newman-- Penrose tetrade (\cite{NP}) $E_{(\mu)}^a \, e^{(\mu)}(x) \equiv E_{(\mu)}^a \, e^{(\mu)}_\mu(x) \, dx^\mu$ on some target space-time equipped with coordinates $x$. This idea possesses all three fundamental constants (gravitational constant reconstructed from $\hbar$, $c$ and the Planck length given by the fuzziness of the space), thus it could point out towards some aspect of the quantum gravity. Precisely this aspect of the construction is the aim of our upcoming analysis.

\newpage
\appendix

\section{Corollaries of Theorem \ref{thm_1}}
\label{sec_ap_A}
\subsection{Connection between the CS $\vert A \rangle$ and standard basis state $\vert n \rangle$}
Recall the definition of the coherent state, cf. (\ref{coherent}), 
\begin{align}
| A \rangle = e^{- |A|^2/2} \sum_{n=0}^\infty \frac{A^n}{n!} | n \rangle \equiv e^{- | A|^2/2} \, \vert A ),\label{coherent_ref} 
\end{align}
where $\vert A ) = e^{|A|^2/2} \, \vert A \rangle \equiv \sum_{n=0}^\infty \frac{A^n}{n!} | n \rangle$. We immediately see that $\partial_A^n \, \vert A ) \, \Big \vert_{A=0} =  \vert n \rangle$, 
thus the desired connection between the two used bases is
\begin{align}
\vert n \rangle =  \partial_A^n \, \vert A ) \, \Big \vert_{A=0} \equiv  \partial_A^n \, \left( e^{|A|^2/2} \, \vert A \rangle \right) \, \Big \vert_{A=0}.\label{connection}
\end{align}
We will use this manipulation to extract the norm of the boosted basis state $\psi_{nm} = e^{\beta \hat{S}}\, \vert n \rangle \langle m \vert$ .

Firstly, denote $\breve{\phi}_{AB} = \vert A ) (B \vert$, then
\begin{align}
\breve{\phi}(\beta) \equiv \breve{\phi}_{AB}(\beta) = e^{\beta \hat{S}} \, \vert A )\, ( B \vert \equiv e^{\left( |A|^2 + |B|^2 \right)/2} \, \phi(\beta).\label{phi_red_AB}
\end{align} 
Our strategy is to calculate the matrix element 
\begin{align}
\langle D \vert \breve{\phi}_{B' A'}(\beta) \vert C \rangle \, \langle C \vert \breve{\phi}_{AB}(\beta) \vert D \rangle,
\end{align}
where primed coherent states are just another general states, i.e. with no connection to (\ref{rescalingAB}). Then (\ref{phiCD}) yields 
\begin{subequations}
\begin{align}
\langle C \vert \breve{\phi}_{AB}(\beta) \vert D \rangle =&\, \frac{1}{\cosh
\frac{\beta}{2}} \, e^{- \tanh \frac{\beta}{2}\, \left( A \bar{B}- \bar{C}D \right) + \frac{1}{\cosh \frac{\beta}{2}} \, \left(A \bar{C}+ \bar{B}D \right)}\, e^{- \left( |C|^2 + |D|^2 \right)/2}, \\
\langle D\vert \breve{\phi}_{B'A'}(\beta) \vert C \rangle =&\, \frac{1}{\cosh
\frac{\beta}{2}} \, e^{- \tanh \frac{\beta}{2}\, \left( \bar{A}' B'- C \bar{D} \right) + \frac{1}{\cosh \frac{\beta}{2}} \, \left(\bar{A}'C+ B'\bar{D} \right)}\, e^{- \left( |C|^2 + |D|^2 \right)/2}
\end{align} 
\end{subequations}
and finally performing the integral over $C,D$ gives rise to
\begin{align}
\mathcal{N}(A,A', B, B') =  \int dC\, dD\, \langle D \vert \breve{\phi}_{B' A'}(\beta) \vert C \rangle \, \langle C \vert \breve{\phi}_{AB}(\beta) \vert D \rangle.
\end{align}

After dividing complex variables into real and imaginary parts one obtains
\begin{align}
... = \frac{1}{\cosh^2 \frac{\beta}{2}}\, e^{\tanh \frac{\beta}{2}\, \left( A \bar{B} +\bar{A}' B'\right)} &\,\int \frac{dC_R\, dD_R}{\pi} \, e^{- C_R^2 - D_R^2 - 2 \tanh \frac{\beta}{2} \, C_R \, D_R + \frac{1}{\cosh \frac{\beta}{2}}\, \left( C_R \, c_1 + D_R \, d_1 \right)} \notag \\
&\, \times \int \frac{dC_I\, dD_I}{\pi} \, e^{- C_I^2 - D_I^2 - 2 \tanh \frac{\beta}{2} \, C_I \, D_I + \frac{1}{\cosh \frac{\beta}{2}}\, \left( C_I \, c_2 + D_I \, d_2 \right)}
\end{align}
with $c_{1,2}$ and $d_{1,2}$ as follows:
\begin{subequations}
\begin{align}
c_1=&\, A_R + A_R' + i\left(A_I - A_I'\right), \\
c_2=&\, i \left(-A_R + A_R' \right) + A_I + A_I', \\
d_1=&\, B_R + B_R'+ i \left(-B_I + B_I'\right), \\
d_2=&\, i \left(B_R - B_R'\right) + B_I + B_I'. 
\end{align}
\end{subequations}
Both double integrals have the same structure yielding $\cosh 
\frac{\beta}{2}\,  \exp \left( \frac{1}{4} \left( c_\alpha^2 - 2 \tanh \frac{\beta}{2} \, c_\alpha d_\alpha + d_\alpha^2 \right) \right) $, where $\alpha=1,2$ is for the real and imaginary part respectively. Substituing for $c_\alpha, d_\alpha$ and combining it with the prefactor $ \frac{1}{\cosh^2 \frac{\beta}{2}}\, e^{\tanh \frac{\beta}{2}\, \left( A \bar{B} +\bar{A}' B'\right)}$ gives the final result $\mathcal{N}(A,A',B,B') = e^{A \bar{A}' + \bar{B} B'}$. We can check that $\mathcal{N}(A,A,B,B) = e^{|A|^2 + |B|^2}$  is exactly the norm of the boosted coherent state $|A ) \, (B |$.  

We are now in position to extract the norm of the boosted state $\psi_{nm}(\beta) = e^{\beta \hat{S}} \, |n \rangle \langle m \vert$. On account of (\ref{connection}): let us produce $\vert n \rangle $ from $| A)$, $\langle m \vert $ from $( B \vert$ and similarly for the primed states. Then the analogous matrix element is 
\begin{equation}
\mathcal{M}(n,n',m,m') = \int dC\, dD\, \langle D \vert\, \psi_{m'n'}(\beta) \,\vert C \rangle \, \langle C \vert \,\psi_{nm}(\beta)\, \vert D \rangle = \partial_A^n \, \partial_{\bar{B}}^m \, \partial_{\bar{A}'}^{n'} \, \partial_{B'}^{m'} \left( e^{A \bar{A}' + \bar{B} B'} \right) \Big \vert_{A= \bar{A}' = \bar{B}= B'=0}.
\end{equation}
We use the binomial theorem in a form directly generalizable to multinomials
\begin{align}
\dots =&\, \partial_A^n \, \partial_{\bar{B}}^m \, \partial_{\bar{A}'}^{n'} \, \partial_{B'}^{m'}\, \sum_{N_1 + N_2=N} \frac{1}{N_1! \, N_2!} \, \left(A \bar{A}'\right)^{N_1}\, \left(\bar{B} B'\right)^{N_2} \, \Big \vert_{A= \bar{A}' = \bar{B}= B'=0} 
\end{align}
and use trivial observation $\partial_x^a \, x^b \Big \vert_{x=0}= a! \, \delta_{ab}$ to get $\mathcal{M}(n,n',m,m') = \delta_{n n'} \, \delta_{m m'}\, n!\, m!$. Thus for $n= n'$ and $m= m'$ we have the result $|| \psi_{nm}(\beta) ||^2_{\mathcal{H}_A} = n!\, m! = || \, \vert n \rangle \langle m \vert \, ||^2_{\mathcal{H}_A}$ we wanted. 

\subsection{Parseval identity among hypergeometric functions}
The procedure of extracting results for the $\vert n \rangle$ basis from the coherent basis described in the previous subsection is an effective tool when computing various matrix elements on $\mathcal{A}$. Let us demonstrate its power for the last time, namely let us compute the following matrix element:
\begin{align}
\mathcal{P}(k,l,m,n) = \langle k \vert \psi_{nm}(\beta) \vert l \rangle \equiv \langle k \vert \left( e^{\beta \hat{S}} \, \vert n \rangle \langle m \vert \right) \vert l \rangle .
\end{align}

We start from $(C \vert \breve{\phi}_{AB}(\beta) \vert D ) = \frac{1}{\cosh \frac{\beta}{2}}\,e^{ \tanh \frac{\beta}{2}\, \left( A \bar{B}- \bar{C}D \right) + \frac{1}{\cosh \frac{\beta}{2}} \, \left(A \bar{C}+ \bar{B}D \right)}$ and extract $\vert n \rangle$ from $\vert A )$, $\langle m \vert$ from $( B \vert$, $\langle k \vert $ from $(C\vert$ and $\vert l \rangle $ from $\vert D \rangle$ in analogy with the previous subsection:
\begin{align}
\mathcal{P}(k,l,m,n) =&\, \frac{1}{\cosh \frac{\beta}{2}}\, \partial_A^n \, \partial_{\bar{B}}^m \, \partial_{\bar{C}}^k  \,\partial_D^l\, \left( e^{\tanh \frac{\beta}{2}\, \left( A \bar{B}- \bar{C}D \right) + \frac{1}{\cosh \frac{\beta}{2}} \, \left(A \bar{C}+ \bar{B}D \right)} \right) \, \Bigg \vert_{A=\bar{B}=\bar{C}=D=0} \notag \\
=&\,  \frac{1}{\cosh \frac{\beta}{2}}\, \partial_A^n \, \partial_{\bar{B}}^m \, \partial_{\bar{C}}^k  \,\partial_D^l\, \sum_{N_1+N_2+N_3+N_4=N} \frac{1}{N_1!\,N_2!\,N_3!\,N_4!}\, \left( \frac{1}{\cosh \frac{\beta}{2}} \, A \bar{C} \right)^{N_1}\notag \\
&\, \left( \frac{1}{\cosh \frac{\beta}{2}} \, \bar{B}D \right)^{N_2} \, \left( \tanh \frac{\beta}{2}\, A \bar{B} \right)^{N_3} \, \left(- \tanh \frac{\beta}{2}\,  \bar{C} D \right)^{N_4} \, \Bigg \vert_{A=\bar{B}=\bar{C}=D=0} \notag \\
=&\,  \frac{1}{\cosh \frac{\beta}{2}}\, \sum_{N_1+N_2+N_3+N_4=N} \frac{1}{N_1!\,N_2!\,N_3!\,N_4!}\, \frac{(-1)^{N_4} \, \tanh^{N_3+N_4} \frac{\beta}{2}}{\cosh^{N_1+N_2}  \frac{\beta}{2}} \notag \\
&\, \partial_A^n \, \partial_{\bar{B}}^m \, \partial_{\bar{C}}^k  \,\partial_D^l\, \left( A^{N_1+N_3} \, \bar{B}^{N_2 + N_3} \, \bar{C}^{N_1 +N_4} \, D^{N_2+N_4} \right) \,  \Bigg \vert_{A=\bar{B}=\bar{C}=D=0}.
\end{align}

At this stage, it is more efficient to write the emergent constraints instead of computing with Kronecker deltas: 
\begin{align}
n=N_1+N_3, \, \, \, m= N_2+N_3, \, \, \, k=N_1+N_4, \, \, \, l= N_2+N_4. 
\end{align}
It is a rank-$3$ system with the constraint $l-k = m-n= \kappa$, cf. (\ref{chirality}). Taking $N_4 = \nu$ as a parameter yields $
N_1= k-\nu$, $N_2=k+ \kappa - \nu$, $N_3= m-k+\nu$. Obviously $N_i \ge 0$ for all $i$, which gives $k \geq \nu$, $\kappa \ge 0$, $\nu \ge k-m $ and $\nu \ge 0$.  It is clear that the range of $\nu$ is $\nu \in \langle \max(0,k-m), k \rangle \equiv D_\nu$, i.e. the resulting matrix element depends on the sign of $k-m$.

Using all previous results we arrive at the following expression for the desired matrix element:
\begin{align}
\dots =&\, \frac{1}{\cosh \frac{\beta}{2}}\,\sum_{\nu \in D_\nu} \frac{(-1)^\nu \, \tanh^{m-k+2\nu} \frac{\beta}{2} }{\cosh^{2k+\kappa-2\nu} \frac{\beta}{2}}\, \frac{k!\, (m-\kappa)! \,(k+\kappa)!\, m!}{(k-\nu)! \,(k+\kappa-\nu)! \,(m-k+\nu)! \,\nu!}\notag \\
=&\, \frac{\tanh^{m-k} \frac{\beta}{2}}{\cosh^{2k+\kappa+1} \frac{\beta}{2}}\, \sum_{\nu \in D_\nu} \frac{(-1)^\nu \sinh^{2\nu} \frac{\beta}{2} \, k!\, (m-\kappa)! \,(k+\kappa)!\, m! }{(k-\nu)! \,(k+\kappa-\nu)! \,(m-k+\nu)! \,\nu!} \notag \\
=&\, \frac{\tanh^{m-k} \frac{\beta}{2}}{\cosh^{2k+\kappa+1} \frac{\beta}{2}} \cdot \begin{cases} 
\frac{m!\,(m-\kappa)! }{(m-k)!}\, _2F_1(-k,-k-\kappa;-k+m+1; - \sinh^2  \frac{\beta}{2}), & k \le m, \\
\frac{(-\sinh^2  \frac{\beta}{2})^{k-m}\,k!\,(k+\kappa)!\,(m-\kappa)! }{(k-m)!\,(m+\kappa)!}\, _2F_1(-m,-m-\kappa;k-m+1; - \sinh^2  \frac{\beta}{2}), & k>m.
\end{cases} 
\end{align}
or 
\begin{align}
\mathcal{P}(k,l,m,n) = \frac{\tanh^{m-k} \frac{\beta}{2}}{\cosh^{k+l+1} \frac{\beta}{2}} \cdot \begin{cases} 
\frac{m!\,n! }{(m-k)!}\, _2F_1(-k,-l;-k+m+1; - \sinh^2  \frac{\beta}{2}), & k \le m, \\
\frac{(-\sinh^2  \frac{\beta}{2})^{k-m}\,k!\,l!\,n! }{(k-m)!\,(2m-n)!}\, _2F_1(-m,-(2m-n);k-m+1; - \sinh^2  \frac{\beta}{2}), & k>m.
\end{cases} 
\end{align}

Let us now use this result to prove an identity among hypergeometric functions, which seems not to be mentioned anywhere. Recall
\begin{align}
|| \psi_{nm}(\beta) ||^2_{\mathcal{H}_A} = \sum_{k,l} \vert   \langle k \vert \psi_{nm}(\beta) \vert \, l \rangle \vert^2 \equiv \sum_{k,l} \mathcal{P}(k,l,m,n)^2 = n!\,m!.
\end{align}
Then substituting for $\mathcal{P}$  and dividing by $n!\,m!$ completes the proof of the following

\begin{prop}[Parseval identity for hypergeometric functions] 
\label{parseval}
\begin{align}
1=&\, \sum_{l=0}^\infty \Bigg[ \sum_{k=0}^m  n!\,m! \, \left( \frac{\tanh^{m-k} \frac{\beta}{2} }{(m-k)!} \right)^2 \frac{_2F_1(-k,-l;-k+m+1, -\sinh^2 \frac{\beta}{2} )^2}{\cosh^{2(k+l+1)} \frac{\beta}{2} } \notag \\
&\, \, \, \, \, \, \, \, \, \, \, +\sum_{k=m+1}^\infty  \frac{k!^2\,l!^2 \, n!}{m!\,(2m-n)!^2} \, \left( \frac{\tanh^{k-m} \frac{\beta}{2} }{(k-m)!} \right)^2 \frac{_2F_1(-m,-(2m-n);k-m+1, -\sinh^2 \frac{\beta}{2} )^2}{\cosh^{2(n+m+1)} \frac{\beta}{2} } \Bigg]
\end{align}
for any $n,m \in \mathbb{N}_0$ and $\beta\in \mathbb{R}$.
\end{prop}

\end{document}